\documentclass[onecolumn]{IEEEtran}
\usepackage{amssymb,epsfig,graphics, graphicx, url,times,mathrsfs,algorithm,algorithmic,amsmath,array,latexsym,fancyhdr,xspace,wrapfig, xcolor,multirow,amsthm,caption,cite,helvet,sidecap,bm,hyperref,url,subcaption}

\usepackage[normalem]{ulem}

\usepackage[inline]{enumitem}
\usepackage{arydshln}
\usepackage[makeroom]{cancel}
\usepackage{varwidth}
\usepackage{bbm}

\usepackage{tikz}
\usetikzlibrary{shapes.geometric}
\usetikzlibrary{shapes,snakes,patterns}
\usetikzlibrary{arrows,positioning,automata,calc}
\usetikzlibrary{plotmarks}

\definecolor{aqua}{rgb}{0.0, 1.0, 1.0}
\definecolor{babypink}{rgb}{0.96, 0.76, 0.76}
\definecolor{brightgreen}{rgb}{0.4, 1.0, 0.0}
\definecolor{brightpink}{rgb}{1.0, 0.0, 0.5}
\definecolor{cadmiumgreen}{rgb}{0.0, 0.42, 0.24}
\definecolor{carrotorange}{rgb}{0.93, 0.57, 0.13}
\definecolor{green1}{rgb}{0,0.5,0}
\definecolor{magenta}{rgb}{1.0, 0.11, 0.81}
\definecolor{mulberry}{rgb}{0.77, 0.29, 0.55}
\definecolor{smalt}{rgb}{0.0, 0.2, 0.6}
\definecolor{xgray}{rgb}{0.9, 0.9, 0.9}

\def \bes{\begin{equation*}}
\def \ees{\end{equation*}}
\def \bas{\begin{align*}}
\def \eas{\end{align*}}
\def \be{\begin{equation}}
\def \ee{\end{equation}}
\def \bbm{\begin{bmatrix}}
\def \ebm{\end{bmatrix}}

\newcommand{\cC}{{\cal C}}
\newcommand{\cD}{{\cal D}}

\newcommand{\cM}{{\cal M}}

\newtheorem{theorem}{Theorem}
\newtheorem{corollary}{Corollary}
\newtheorem{lemma}{Lemma}

\newtheorem{definition}{Definition}

\newtheorem{remark}{Remark}

\usepackage[font=small]{caption}

\newcommand{\capacity}{\texttt{C}}
\newcommand{\cut}{\texttt{C}}
\newcommand{\rate}{\texttt{R}}

\usetikzlibrary{decorations.pathreplacing,calc}
\tikzset{brace/.style={decorate, decoration={brace}},
 brace mirrored/.style={decorate, decoration={brace,mirror}},
}

\newcounter{brace}
\IEEEoverridecommandlockouts

\begin{document}

\title{Network Coding with Myopic Adversaries}

\author{
\IEEEauthorblockN{
Sijie~Li\thanks{Sijie Li is with the Department of Information Engineering, The Chinese University of Hong Kong, Shatin, Hong Kong. Email: \texttt{sijieli@link.cuhk.edu.hk}.},
Rawad~Bitar\thanks{Rawad Bitar is with the Institute for Communications Engineering, Technical University of Munich, Munich, Germany. Email: \texttt{rawad.bitar@tum.de}.},
Sidharth~Jaggi\thanks{Sidharth Jaggi is with the School of Mathematics, University of Bristol, Bristol, United Kingdom, and the Department of Information Engineering, The Chinese University of Hong Kong, Shatin, Hong Kong. Email: \texttt{sid.jaggi@bristol.ac.uk}.}
and~Yihan Zhang\thanks{Yihan Zhang is with the Faculty of Computer Science, Technion Israel Institute of Technology, Haifa, Israel. Email: \texttt{yihanzhang@cuhk.edu.hk}.}\\
}


}


\maketitle
\begin{abstract}
We consider the problem of reliable communication over a network containing a hidden {\it myopic} adversary who can eavesdrop on some $z_{ro}$ links, jam some $z_{wo}$ links, and do both on some $z_{rw}$ links. We provide the first information-theoretically tight characterization of the optimal rate of communication possible under all possible settings of the tuple $(z_{ro},z_{wo},z_{rw})$ by providing a novel coding scheme/analysis for a subset of parameter regimes. In particular, our vanishing-error schemes bypass the Network Singleton Bound
(which requires a zero-error recovery criteria) in a certain parameter regime where the capacity had been heretofore open. As a direct corollary we also obtain the capacity of the corresponding problem where information-theoretic secrecy against eavesdropping is required in addition to reliable communication.
\end{abstract}


\begin{center}
{\small{\textbf{\color{red} A short video describing this work  can be found in~\cite{video}.}}}
\end{center}

\section{Introduction}
\label{sec:intro}

Network coding is a network communication paradigm wherein nodes in a network non-trivially combine incoming packets to generate information on outgoing packets. It has been shown~ \cite{ahlswede2000network} that such combination operations are necessary and sufficient to attain information-theoretically optimal communication rates for many classes of network communication problems -- in particular, for {\it multicast} problems, if the smallest source-destination {\it min-cut} equals $\capacity$, network codes are able to attain this rate. Further, it was shown\cite{li2003linear,koetter2003algebraic,jaggi2005polynomial} that linear codes suffice for this purpose. Applications of network coding now abound in a variety of disparate fields, such as wireless systems~\cite{zhang2006hot}, distributed storage~\cite{dimakis2010network}, and router designs~\cite{CodedBulk:21}.

One complication in the network coding paradigm is the potential problem of errors -- due to the mixing operations in the network even a single corrupted packet may end up tainting the majority of the information flowing in the network; as such, a malicious jammer hiding in the network may be able inflict disproportionate damage. To combat this problem, network error-correcting codes were proposed by Cai and Yeung~\cite{yeung2006network,cai2006network}, followed by a plethora of computationally efficient code designs~\cite{JagLKHME:08,silva2011universal,silva2010universal,metric,KK08}.

The focus of this paper is on a complete characterization of the optimal throughput possible in the presence of a {\it myopic} jammer. Initial works on network error-correction (for instance~\cite{yeung2006network,cai2006network,KK08}) assumed the presence of an {\it omniscient} adversary -- an adversary who is able to observe {\it all} transmissions in a network, and then tailor his jamming scheme as a function of his observations. In such scenarios, it was shown as a consequence of the Network Singleton bound ~\cite{yeung2006network,cai2006network} that each of the adversary's injected corrupted packets can do ``double damage'', i.e., the optimal throughput obtainable is $\capacity-2z_{w}$, 
where $z_w$ equals the number of packets the adversary can inject into the network. In contrast, it was observed in~\cite{JagLKHME:08} that if the adversary is able to observe only $z_r$ packets and must design its jamming strategy as a function of these observations, then for the parameter regime $z_r+2z_{w} < \capacity$ (whence the adversary was said to be {\it limited}) in fact a throughput of $\capacity-z_w$ is obtainable, effectively making the jamming no more damaging than the relatively benign scenario of random noise.\footnote{It is important to highlight that such results attaining rates higher than $\capacity-2z_w$ are possible only if one accepts a vanishing probability of error metric rather than a zero-probability of error metric.} Extensions beyond this parameter regime were made in the setting where the adversary's noise is {\it additive}~\cite{YaoSJL:14,hayashi2017secrecy}. Such adversarial models arise naturally in a variety of settings wherein the adversary can only control (eavesdrop on and/or jam) a subset of network links due to its physical constraints. However, a complete information-theoretic characterization of the capacity region, especially for the important and physically relevant model of {\it overwrite} adversaries (see Remark~\ref{rem:add-ow}) was heretofore open.

In this work we focus on a more granular model parametrization that subsumes the limited adversary model of~\cite{JagLKHME:08} as a special case. For this generalized setting we provide a complete characterization of the information-theoretically optimal communication rate possible. In our model, there are:
\begin{itemize}
    \item {\underline{Read-Only links}}: $z_{ro}$ links that the adversary can only observe but not jam.
    \item {\underline{Write-Only links}}: $z_{wo}$ links that the adversary can only jam but not observe.
    \item {\underline{Read-Write links}}: $z_{rw}$ links that the adversary can both observe and jam.
\end{itemize}
As ancillary parameters, we also define
\begin{itemize}
    \item {\underline{Read links}}: $z_r$ denotes the overall number of links $z_{ro} + z_{rw}$ that the adversary can observe.
    \item {\underline{Write links}}: $z_w$ denotes the overall number of links $z_{wo} + z_{rw}$ that the adversary can jam.
\end{itemize}

Our main result is that if $z_{ro} + 2z_{rw} + 2z_{wo} < \capacity$ (the so-called ``weak adversary'' regime), then a rate of  $\capacity-z_w$ is attainable. The optimality of this rate can be seen by noting that if James were to just add random noise on $z_w$ links in the min-cut, direct information-theoretic cutset arguments imply that no higher rate is possible without resulting in a probability of error approaching one. Our results are quite strong -- they hold even in a distributed network coding setting, i.e., if none of the legitimate parties communicating have prior information neither about the network topology or linear network coding operations performed by internal nodes nor about which network links are being eavesdropped/jammed by the adversary. In contrast, the malicious adversary is assumed to know the network topology, coding operations of each network node (including the source and sink) in advance, and as a function of this knowledge is allowed to choose an arbitrary subset of $z_r$ links to eavesdrop on. On a basis of these observations the adversary may also choose $z_w$ links to jam (of which at most $z_{rw}$ may be from among the $z_r$ eavesdropped links), and additionally may base the contents of the corrupted packets he injects on all this information. Also, no computational restrictions are assumed on the adversary.

On the other hand if $z_{ro} + 2z_{rw} + 2z_{wo} \geq \capacity$ (the so-called ``strong adversary'' regime), prior work~\cite{converse} has already shown that even in particularly simple networks (``parallel-edge networks") no rate higher than $(\capacity-2z_w)^{+}$ is attainable, and indeed such a rate is already obtainable even against omniscient adversaries, for instance by the codes in~\cite{JagLKHME:08,silva2011universal,silva2010universal,metric,KK08}. Indeed, our main result may be viewed as the network coding generalization of~\cite{converse}.
A comparison of our work with related prior works is listed in Table~\ref{tab:related-work}. 
\begin{table*}[htbp]
\centering
\begin{tabular}{|l|l|l|l|l|}
		\hline
		 \textbf{} & \textbf{Adversary power} & \textbf{Network type} &\textbf{Rate} & \textbf{Decoding Complexity} \\ \hline
        Jaggi et.al. \cite{JagLKHME:08} & Strong Adv. $\capacity< z_{ro}+ 2z_{rw}+2z_{wo}$ & General Network & $\capacity-2z_w$ & $\mathcal{O}((n\capacity)^3)$\\ \hline
        Jaggi et.al. \cite{JagLKHME:08} & Weak Adv. $\capacity>z_r + 2z_w$ & General Network & $\capacity-z_w$ &  $\mathcal{O}(n\capacity^2)$\\ \hline
        Zhang et.al. \cite{converse} & Weak Adv. $\capacity > z_{ro} + 2z_{rw}+2z_{wo}$ & Parallel Edges & $\capacity-z_w$ & $\mathcal{O}(poly(n))$  \\ \hline
        {\bf This work} & Weak Adv. $\capacity > z_{ro} + 2z_{rw} + 2z_{wo}$ & General Network & $\capacity-z_w$ &  -- \\ \hline        
	\end{tabular}
\caption{\small{\it Related Works. Our work fills the gap of information-theoretical characterization for general network with the optimal rate $\capacity-z_w$. With the converse from \cite{converse}, there now is a full characterization of network error-correction under myopic adversary.}}
\label{tab:related-work}
\end{table*}
 
Our techniques rely on those developed for point-to-point myopic adversarial settings in~\cite{dey2019sufficiently}, carefully coupled with the appropriate {\it subspace metric} for the problem at hand~\cite{metric}. 
In the weak adversary regime, by the myopic nature of the adversary and the choice of the coding rate, the adversary has a considerable amount of uncertainty regarding the codeword transmitted through the network. 
Our analysis critically leverages such uncertainty and shows that under any adversarial action, only a small fraction of codewords may suffer from decoding errors.  

We extend our results to the case where the message must also be secured, in an information theoretic sense, from the adversary's observation. We show that by coupling our coding techniques with coset codes~\cite{el2012secure}, a rate of $\capacity-z_r-z_w$ can be achieved for the weak adversary regime. The optimality of this rate follows from meeting the converse derived in~\cite{converse} for parallel-edge networks. For the strong adversary regime, it is shown~\cite{converse} that no positive rate can be achieved while requiring secrecy of the transmitted message.

 
There are potential applications of such codes in the presence of myopic adversaries in a variety of settings beyond vanilla network coding -- for instance, distributed storage~\cite{pawar2011securing}, secret-sharing~\cite{shah2013secure}, private information retrieval~\cite{wang2018epsilon}, and coded computing~\cite{yu2019lagrange}.
 
 


\section{Preliminaries}\label{sec:Model and Subspace code}

For ease of presentation, in this paper we consider a unicast network coding problem in the presence of a myopic adversary -- as is common in the network error-correction literature, the techniques we develop directly translate to multicast settings as well.

\subsection{Channel Model}

\label{sec:model}

We consider the problem of communicating reliably through a network in the presence of a myopic adversary. In a nutshell, a sender Alice wants to send a message to a receiver Bob through the network. An adversary James eavesdrops on a subset of the links in the network and can jam another subset of the links -- in particular, James can decide which subset of links to jam, and how to jam them, based on his observations from the eavesdropped links. The goal is for Bob to be able to reliably reconstruct Alice's message despite James' jamming action. 
The network model is depicted in Figure~\ref{fig:problem_model}. 
\begin{figure}[htbp]
    \centering
    \begin{tikzpicture}[node distance=2.3cm]
         
    \node (X)  {$X$};
    \node[rectangle, minimum width = 2cm] (network) [right = 1.5cm of X] {Network};
    \node (Y) [right = 1 cm of network] {$Y$};
    \node (X1) [right = 1.2 cm of Y] {$\hat X$};
    \node (Z) [below = 1.5 of network.south west] {$Z$};
    \node (S) [right of = Z] {$S$};
    
    \path[->]
    ([xshift = 0.5cm] network.south west) edge node[anchor=east, font = \footnotesize] {$\substack{\mbox{Adversary's}\\ \mbox{observation}}$} (Z)
    (S) edge node[anchor=west, font = \footnotesize] {$\substack{\mbox{Corrupted}\\ \mbox{packets injected}}$} ([xshift = -0.3cm] network.south east)
    (Z) edge node[anchor=south, font = \footnotesize] {Choice of} node[anchor=south, font = \footnotesize, below] {jamming packets} (S);
    
    \draw [->] (X) -- node[anchor=south, font = \footnotesize] {Encoder} (network);
    \draw [->] (network) --  (Y);
    \draw [->] (Y) -- node[anchor=south, font = \footnotesize] {Decoder} (X1);
    \end{tikzpicture}
    \caption{\it \small {}Network Model: 
    The sender inputs packets corresponding to the rows of  matrix $X$ into the network. The adversary eavesdrops on a set of $z_r$ links leading to his observation matrix $ Z $.
    Based on $Z$, he jams another set of links by injecting into the network a matrix $ S $. 
    The receiver receives the corrupted packets $Y$, corresponding to linear combinations of $X$ and $S$. The communication goal is for the receiver to recover the transmitted $X$ with high probability.} 
    \label{fig:problem_model}
\end{figure}
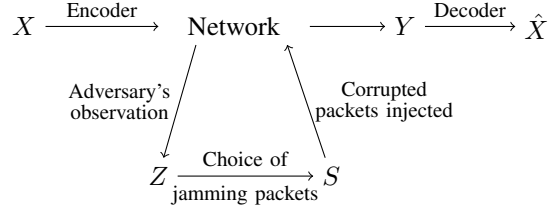
The detailed model is explained next.


\noindent {\bf Notational conventions:}
We use $\mathbb{F}_q$ to denote finite fields of size $q$ for prime powers $q$, and $\mathbb{F}_q^n$ or $(\mathbb{F}_q)^n$ to denote the vector space of $n$-tuples over $\mathbb{F}_q$. Scalars and scalar functions will be denoted by lower-case alphabets (e.g. $m$). Matrices will be denoted by upper-case alphabets (e.g. $X$) -- two exceptions as nods to firmly established convention will be the scalar quantity $\capacity$ denoting the min-cut of a graph of interest, and $U(\cdot)$ denoting the uniform distribution over a set. The row-space of any given matrix $X$ will be denoted $V(X)$, and $\oplus$ and $\cap$ respectively denote the direct sum and intersection (of vector spaces). Sets and graphs will be denoted by calligraphic letters (e.g. ${\cal C}$). The notation $\mathbbm{1}(.)$ indicates the indicator function of the corresponding event, and the notation $(x)^+$ denotes $\max\{x,0\}$.
Frequently used notation throughout this paper is listed in Table~\ref{tab:notation} for reference. 
\begin{table}[htbp]
\renewcommand{\arraystretch}{1.2}
\centering
	\begin{tabular}{|l|l|}
		\hline
		\textbf{Symbol} & \textbf{Meaning}\\ \hline
		$\capacity$   &        Min-cut of the Network         \\ \hline
		$n$   &        Length of the packet             \\ \hline
		$q$   &        Size of the finite field         \\ \hline
		$\rate$   &        Rate         \\ \hline
		$z_{ro}$ &        Eavesdropping-only power         \\ \hline
		$z_{rw}$ &        Eavesdropping and Overwriting power         \\ \hline
		$z_{wo}$ &        Overwriting-only power         \\ \hline
		$z_r$ &        Total eavesdropping power  $=z_{ro}+z_{rw}$       \\ \hline
		$z_w$ &        Total overwriting power $=z_{wo}+z_{rw}$         \\ \hline
		$\epsilon$   &    Rate slack -- small positive constant         \\ \hline

	    $\mathcal{C}$   &        Codebook             \\ \hline
	\end{tabular}
	\caption{Important notation }
	\label{tab:notation}
\end{table}


\noindent {\bf Network model:} 
The network ${\cal N}$ is a directed acyclic graph\footnote{The scenario where the network has cycles is considerably more complex~\cite{yang2013network}, as it involves some level of feedback -- we do not consider it here.} comprising of a {\it vertex-set} ${\cal V}$ of nodes and an {\it edge-set} ${\cal E}$ of directed links. Each node in ${\cal V}$ can manipulate information on incoming links to generate messages on outgoing links. Each link in ${\cal E}$ is assumed to have unit capacity, corresponding to the ability to transmit a single length-$n$ vector over some finite field ${\mathbb{F}}_q$ over a suitable period of time.~\footnote{Again, as is standard in the network coding, if links have unequal capacities, this can be handled by splitting such links into parallel links of unit capacity.} The {\it block/packet-length} $n$,  and {\it field-size} $q$ are design parameters to be specified later. Nodes in the network may perform arbitrary arithmetic operations on incoming packets to generate outgoing packets. The min-cut of the network is denoted by $\capacity$.

\noindent {\bf Encoder model:}
The sender Alice, situated at the {\it source node}, has a {\it message} $m$ that is uniformly distributed over the set $\left [ q^{n\rate} \right ]$  (the {\it rate} $\rate$ is a design parameter to be specified later).

Alice's goal is to reliably communicate her message $m $ to the receiver Bob situated at the {\it sink node}. To instantiate this communication she uses her {\it encoder} $Enc:  \left [ q^{n\rate}\right ]\rightarrow ({\mathbb F}_q)^{\capacity \times n}$ to choose a  {\it codeword} $X$ (a $\capacity \times n$ matrix over ${\mathbb{F}}_q$) for each message $m \in \left [ q^{n\rate} \right ]$.\footnote{Note that in contrast to some prior work in the secure/reliable network coding literature (for instance~\cite{JagLKHME:08}), this is a {\it deterministic} mapping from each message $m$ to corresponding codeword $X$ -- it turns out that our schemes do not need to rely on additional stochasticity/randomness.} The collection of all such codewords $X$ comprises the {\it codebook} ${\cal C}$. Prior to communication, this codebook ${\cal C}$ and the corresponding encoding (mapping from messages $m$ to codewords $X$) is known to each of Alice, Bob, and the adversary James described below.

We assume that prior to communication, Alice, Bob and the intermediate nodes do not know the network topology (though they know the value of the min-cut $\capacity$), nor do they know the network coding operations performed by intermediate nodes.\footnote{We make these model choices to demonstrate that our coding scheme is able to operate despite knowing very little about the network setting {\it a priori}. In addition, the converse argument we outline in Section~\ref{sec:converse} goes through even if Alice/Bob/intermediate nodes had prior knowledge of the topology/intermediate coding operations, so such an assumption is not unduly restrictive.} In the scheme we present these intermediate nodes perform {\it random linear network coding}~\cite{ho2006random}, though the matching converse argument we outline in Section~\ref{sec:converse} does not rely on specific coding scheme.

\noindent  {\bf Network communication scheme:}
The {\it network communication scheme} ${\cal S}$ comprises of Alice's encoder $Enc$, the  (linear) coding operations of nodes in ${\cal V}$, and Bob's decoder $Dec$ as described below.

\noindent  {\bf Adversarial model:} In addition to knowing Alice's encoding strategy/codebook ${\cal C}$, the malicious adversary James knows the network topology, Bob's decoding strategy, and the coding performed at intermediate nodes. James' goal is to try to disrupt the communication from Alice to Bob in a manner so that Bob is unable to reliably estimate Alice's message $m$.
To instantiate this disruption, as a function of his knowledge, he can pick a subset of  
links to control in the manner described below. 
James's power is parametrized by his {\it adversarial power-tuple} $(z_{ro},z_{wo},z_{rw})$, characterized as the following: 
\begin{itemize}
    \item First, he can read (without changing) the data transmitted on the set of $z_{ro}$ ``read-only'' links of his choice.
    \item Next, on another set of $z_{rw}$ ``read-write'' links of his choice, he can read the transmitted data, and then overwrite the transmissions on these links with an arbitrary set of $z_{rw}$ length-$n$ vectors (these vectors may depends on James' observations on all $z_{ro}+z_{rw}$ links).
    \item Finally, on a set of $z_{wo}$ ``write-only'' links of his choice, James can replace the contents of these links with an arbitrary set of $z_{wo}$ length-$n$ vectors (these vectors may depends on James' observations on all $z_{ro}+z_{rw}$ links, but {
\it not} on the contents of the $z_{wo}$ write-only links).
\end{itemize}
For notational convenience, we also define the ancillary parameters $z_r$, $z_w$ and $z$ as follows:
\begin{itemize}
    \item The {\it number of eavesdropped links} $z_r$ is set to equal $z_{ro}+z_{rw}$, corresponding to the total number of links James can eavesdrop on.
    \item The {\it number of jammed links} $z_w$ is set to equal $z_{wo}+z_{rw}$, corresponding to the total number of links James can jam.
    \item The {\it number of corrupted links} $z$ is set to equal $z_{ro}+z_{wo}+z_{rw}$, corresponding to the total number of links James can read and/or write on.
\end{itemize}

In more detail, let the {\it eavesdropper's observation matrix} $Z$ be the $z_{r}\times n$ matrix over ${\mathbb{ F}}_q$ whose rows comprise of James' observations on the $z_{r}$ links he can eavesdrop on. As noted above, since in this model it suffices to restrict the operations performed by intermediate notes to (random) linear network coding operations, $Z$ equals $T_{AJ}X$. Here the {\it network transform from Alice to James} $T_{AJ}$ is a $z_{r}\times \capacity$ matrix over ${\mathbb{ F}}_q$ corresponding to the linear  transform of $X$ instantiated by the network coding operations by nodes upstream of James.  

Further, let the {\it jamming matrix} $S$ be the $z_{w}\times n$ matrix over ${\mathbb{ F}}_q$ which comprises of James' jamming patterns $ S_{wo}\in\mathbb{F}_q^{z_{wo}\times n} $ on the $ z_{wo} $ links he can jam and $ S_{rw}\in\mathbb{F}_q^{z_{rw}\times n} $ on the $z_{rw}$ links he can eavesdrop on and jam. Then $S$ is a function of $Z$ (and in addition James' knowledge of the network topology, and Alice's encoder $Enc$ above, the network coding operations, and Bob's decoder $Dec$ described below). On a few occasions below we use the notation $S=Jam_{\cal S}(Z)$ instead of $S$, to make explicit the dependence of the jamming matrix $S$ on the eavesdropper's observation matrix $Z$ -- here $Jam_{\cal S} : ({\mathbb F}_q)^{z_r \times n} \rightarrow ({\mathbb F}_q)^{z_w \times n}$ can be interpreted as James' {\it jamming function}. James is unconstrained in his choice of jamming functions.\footnote{Indeed, he can even choose probabilistic jamming functions. However, as shown in~\cite{dey2019sufficiently}, given any probabilistic jamming function $Jam_{\cal S}$ with a given probability of decoding error (as defined in~\eqref{eq:prob-error}), there exists a deterministic jamming function $Jam'_{\cal S}$ with at least the same probability of decoding error. Hence, without loss of generality, we focus here on deterministic jamming functions.}

\begin{remark}
Strictly speaking, the jamming matrix $ S\in\mathbb{F}_q^{z_w'\times n} $ can have smaller dimension $ z_w'\le z_w $. 
In particular, it can comprise of two sub-matrices $ S_{wo}\in\mathbb{F}_q^{z_{wo}'\times n} $ and $ S_{rw}\in\mathbb{F}_q^{z_{rw}'\times n} $ for some $ z_{wo}'\le z_{wo} $ and $ z_{rw}'\le z_{rw} $.
However, we focus on the worst case where the adversary uses his full power. 
Other cases can be reduced to this case by treating $ (z_{ro}',z_{wo}',z_{rw}') $ that James truly used as the new adversarial power tuple.
\end{remark}

Prior to communication, the locations of these $z_{ro}$, $z_{wo}$, and $z_{rw}$ links among the edge-set ${\cal E}$ are unknown to Alice/Bob/intermediate nodes (though the {\it values} of $z_{ro}$, $z_{wo}$, and $z_{rw}$, or good upper bounds on these, are available to Alice and Bob).


Following the lead of~\cite{converse}, James is called a {\it weak adversary} if the condition in Eqn.~\eqref{eq:weak-adversary} is satisfied.
\begin{equation}
\capacity>z_{ro}+2z_{w}\label{eq:weak-adversary}
\end{equation}

A pictorial explanation of the adversarial model considered in this work is shown in Fig.~\ref{fig:adversary_network}.
\begin{figure}[htbp]
    \centering
    \begin{tikzpicture}
    \newcommand{\asymcloud}[2][.1]{%
\begin{scope}[#2]
\pgftransformscale{#1}%
\pgfpathmoveto{\pgfpoint{261 pt}{115 pt}} 
  \pgfpathcurveto{\pgfqpoint{70 pt}{107 pt}}
                 {\pgfqpoint{137 pt}{291 pt}}
                 {\pgfqpoint{260 pt}{273 pt}} 
  \pgfpathcurveto{\pgfqpoint{78 pt}{382 pt}}
                 {\pgfqpoint{381 pt}{445 pt}}
                 {\pgfqpoint{412 pt}{410 pt}}
  \pgfpathcurveto{\pgfqpoint{577 pt}{587 pt}}
                 {\pgfqpoint{698 pt}{488 pt}}
                 {\pgfqpoint{685 pt}{366 pt}}
  \pgfpathcurveto{\pgfqpoint{840 pt}{192 pt}}
                 {\pgfqpoint{610 pt}{157 pt}}
                 {\pgfqpoint{610 pt}{157 pt}}
  \pgfpathcurveto{\pgfqpoint{531 pt}{39 pt}}
                 {\pgfqpoint{298 pt}{51 pt}}
                 {\pgfqpoint{261 pt}{115 pt}}
\pgfusepath{fill,stroke}         
\end{scope}}  

\node (cloud) at (0,0) {\tikz \asymcloud[.25]{fill=white!20,thick};};

\node at (0,-1.7) {Network};
\node (james) at (-0.5,2) {James};
\filldraw[black](-3,0) circle (2.5pt) node[anchor=east,left = 0.1 cm] {Alice};
\coordinate (alice) at (-3,0);

\node[circle,radius = 2pt, draw = black] (l1) at (-1,0.5) {};
\node at (-1,0.15) {$\vdots$};
\node[circle,radius = 2pt, draw = black] (l2) at (-1,-0.4) {};
\node[circle,radius = 2pt, draw = black] (l3) at (-1,-1) {};
\node[circle,radius = 2pt, draw = black] (l4) at (-1,-1.6) {};

\node[circle,radius = 2pt, draw = black] (21) at (0,0.5) {};
\node at (0,0.15) {$\vdots$};
\node[circle,radius = 2pt, draw = black] (22) at (0,-0.4) {};
\node[circle,radius = 2pt, draw = black] (23) at (0,-1) {};

\node (dot) at (0.6,0) {$\vdots$};

\node[circle,radius = 2pt, draw = black] (31) at (1.2,0.8) {};
\node at (1.2,0.35) {$\vdots$};
\node[circle,radius = 2pt, draw = black] (32) at (1.2,-0.2) {};
\node[circle,radius = 2pt, draw = black] (33) at (1.4,-0.7) {};
\node[circle,radius = 2pt, draw = black] (34) at (1.8,-1) {};

\filldraw[black] (3,0.8) circle (2.5pt) node[anchor=west,right = 0.1 cm] {Bob};
\coordinate (bob) at (3,0.8);

\path [->]
(alice) edge[color = blue] (l1)
        edge (l2)
        edge (l3)
        edge (l4)
(l1)    edge (21)
        edge (22)
(l2)    edge[color = blue] (21)
        edge[color = green!50!black] (23)
(l3)   
        edge[color = green!50!black] (22)
        edge (23)
(l4)    edge (21)
        edge (22)
        edge (23)
(21)   edge (dot)
(22)   edge (dot)
(23)   edge[color = red] (dot)
(dot)   edge (31)
        edge (32)
        edge (33)
        edge [bend right] (34)
(31)   edge (bob)
(32)   edge (bob)
(33)   edge[color = red] (bob)
(34)   edge (bob);

\draw[very thick, ->] (james) -- (-0.5,1);



    \end{tikzpicture}
    \caption{\it \small{ James' attack in view of the network. The {\color{blue}blue} edges represent those that James can eavesdrop. The {\color{red}red} edges represent those that James can jam. The {\color{green!50!black}green} edges represent those that James can both eavesdrop and jam. }}
    \label{fig:adversary_network}
\end{figure}
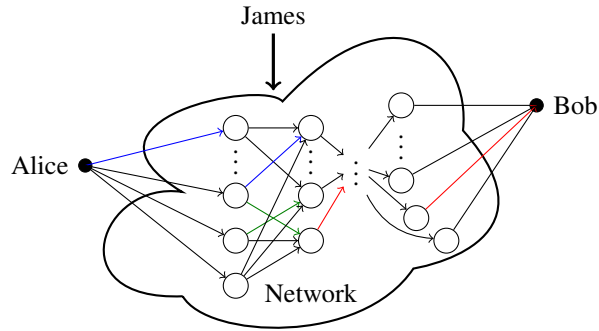

\begin{remark}
As noted in Section~\ref{sec:intro}, the significance of the inequality in Eqn.~\eqref{eq:weak-adversary} is that this is precisely the parameter regime where the following happens -- say Alice chooses a random code of rate $\capacity-z_w-
\epsilon$, then there is still an exponentially large
set of codewords in ${\cal C}$ that are consistent with James' observation $Z$. Roughly speaking, in~\cite{dey2019sufficiently}, whose approach we adapt in this work, the corresponding parameter regime is called the {\it sufficiently myopic} regime.
\end{remark}

\begin{remark}~\label{rem:add-ow} The distinction between the additive error models in ~\cite{YaoSJL:14,hayashi2017secrecy}, and the {\it overwrite} model considered in this work (wherein James can replace the contents of packets on the links he can jam with whatever he wishes) shows itself in the $z_{wo}$ links. In additive models, if James has uncertainty about what is being transmitted on a link, he will still have uncertainty after he jams this link. In contrast, in the overwrite model we consider in this work, the content of packets on links James corrupts is always precisely known to him, since he replaces the prior contents with his injected corruptions (even if he has uncertainty about the contents of the links he is corrupting). Arguably, the overwrite error model of this paper is a more natural fit for a variety of wired/distributed computing/storage models than the additive model (which can perhaps be motivated more in wireless settings).
\end{remark}


\begin{remark}
We wish to emphasize that anything that Alice and Bob (and/or intermediate nodes) know prior to communication, James also knows -- hence no {\it shared keys/common randomness} is shared privately between  Alice  and  Bob -- in this regard we differ from some models in the literature, such as the ``Shared Secret'' model in~\cite{JagLKHME:08}, or the model of~\cite{tian2016arbitrarily}. Also, we do not assume computational bounds on James (unlike, for instance, the models of~\cite{charles2006signatures} or~\cite{zhao2007signatures}).
\end{remark} 

\noindent {\bf Decoder model:} We represent the information on the links incoming to the sink by the {\it network output} $Y$, a $\capacity\times n$ matrix over ${\mathbb F}_q$.~\footnote{A natural question pertains to scenarios where there are more than $\capacity$ packets incoming to the sink. It can be shown via standard arguments that 
with high probability over the random linear network code design, there are at least $\capacity$ linearly independent vectors on the links incoming to the sink. As is common in the network error-correction literature (see for instance~\cite{JagLKHME:08}), if there are more than $\capacity$ linearly independent vectors, we choose an arbitrary subset of size $\capacity$ and discard the remainder. As we show, in the weak adversary regime when Alice is transmitting at rate $\capacity-z_w-\epsilon$, Bob is still able to reconstruct Alice's message with high probability. Conversely, via standard information-theoretic arguments, if Alice is transmitting at rate higher than $\capacity-z_w+\epsilon$ and James injects random noise on $z_w$ links situated in a min-cut, {\it every} communication scheme will have a probability of error converging to $1$. Hence no loss of performance arises from this discarding operation. 
} Given this $Y$ and his knowledge of Alice's codebook ${\cal C}$, the goal of Bob's decoder $Dec: ({\mathbb F}_q)^{\capacity \times n} \rightarrow \left [ q^{n\rate}\right ]$  is to ensure that its output $\hat{m}$ is a ``reliable estimate" (as made precise next) of Alice's message $m$.

\noindent {\bf Code properties:}
Bob's decoder $Dec$ is said to make an error if the decoder output $\hat{m}$ differs from Alice's message $m$. 

For a given network communication scheme ${\cal S}$, the (average)
{\it probability of decoding error} is defined as in Eqn.~\eqref{eq:prob-error} {where the expectation is over Alice's uniformly distributed message $ m\sim U([q^{nR}]) $ and the random linear network coding operations at the intermediate nodes}.
\begin{equation}
\max_{Jam_{\cal S}}
    \mathop{\mathbb{E}}
    \left ( \frac{\sum_{X' \in {\cal C} : T_{AJ}X' = Z} \mathbbm{1}\left (Dec(Y(X',Jam_{\cal S}(Z))) \neq m \right )}{|\{X' \in {\cal C} : T_{AJ}X' = Z\}|} \right )
\label{eq:prob-error}
\end{equation}

In words, the meaning~\eqref{eq:prob-error} can be unwrapped as follows. Say Alice has message $m$ (the notation $m \sim U\left (\left [ q^{n\rate}\right ]\right )$ means that $m$ is uniformly distributed among all possible messages), resulting in the codeword $X = Enc(m)$. For the given network communication scheme ${\cal S}$ (that James knows) his observation matrix equals $Z = T_{AJ}X$, and based on this observation and the communication scheme James chooses a corresponding jamming function $Jam_{\cal S}$, resulting in the jamming matrix $S = Jam_{\cal S}(Z)$. Note that there will in general be multiple possible codewords $X'$  in Alice's codebook ${\cal C}$ such that $T_{AJ}X'$ equals James' observation $Z$ -- call them {\it $Z$-compatible} codewords. Then, for a specific jamming function $Jam_{\cal S}$ and message $m$, the fraction of $Z$-compatible codewords $X'$ that result in Bob's decoder making an error is the probability of error. For a specific jamming function $Jam_{\cal S}$, the average probability of error is the average of the previous quantity over all messages $m$. Finally, since James' jamming function $Jam_{\cal S}$ can be arbitrary (he is after all a malicious adversary), this probability of error quantity is maximized over all possible jamming functions.

A rate $\rate$ is said to be {\it achievable} if for any $\epsilon > 0$ there exists a network communication scheme over some (sufficiently large) $n$ and $q$ such that the probability of decoding error is no more than $\epsilon$. The {\it network error-correction capacity} $\rate^{\ast}$ for a given network ${\cal N}$ and  adversarial power-tuple $(z_{ro},z_{wo},z_{rw})$ is then the supremum (over network communication schemes) of achievable rates.



\noindent {\bf Secrecy model:} 
When secrecy is to be satisfied, our codes attain perfect secrecy \cite{cai2002secure,el2012secure}. Let $m$ be 
the transmitted message, let $X$ be the 
symbols communicated through the network and let $Z$ be 
James' observation. \emph{Information theoretic secrecy} (a.k.a. \emph{perfect secrecy}) requires that James' uncertainty about the message $m$ is not reduced after his observation, i.e., $\textrm{H}(m|Z) = \textrm{H}(m)$, where $\textrm{H}(.)$ is the entropy function and all logarithms are base $q$. This is in contrast to \emph{strong} and \emph{weak secrecy} in which it is required that $\textrm{H}(m|Z) = \textrm{H}(m) - \epsilon_n$ for a small $\epsilon_n$ that either goes to $0$ when the block length $n$ goes to infinity (strong secrecy) or $\epsilon_n / n$ goes to $0$ when $n$ goes to infinity (weak secrecy).

A rate $\rate_{\text{sec}}$ is said to be {\it securely achievable} if for any $\epsilon > 0$ there exists a network communication scheme over some (sufficiently large) $n$ and $q$ such that the probability of decoding error is no more than $\epsilon$ and perfect secrecy of the transmitted message is maintained. The {\it secure network error-correction capacity} $\rate^{\ast}_\text{sec}$ for a given secure network ${\cal N}$ and  adversarial power-tuple $(z_{ro},z_{wo},z_{rw})$ is then the supremum (over network communication schemes) of securely achievable rates.

To prove the strongest possible results, we provide perfect secrecy when constructing codes, and consider weak secrecy for proving a converse on the error-correction capacity of secure networks. We show that those values are equal, i.e., the converse that holds even for weak secrecy can be achieved while maintaining perfect secrecy.

\subsection{Subspace Codes}
In our scheme, Alice's encoder and Bob's decoder will depend critically on certain properties of the row-spaces of the matrices $X$ in ${\cal C}$. It will therefore help to quickly review the extensive literature on {\it subspace codes} (see for instance the review in~\cite{KK08}).

 The set of all subspaces of $\mathbb{F}_q^n$ is called the projective space 
of order $n$ over $\mathbb{F}_q$, denoted as $\mathcal{P}_q (n)$. The set of all $k$-dimensional subspaces of $\mathbb{F}_q^n $ is called a Grassmannian, denoted as $\mathcal{G}_q (n,k)$, where $0 \leq k \leq n$. 
A graph representation of the Grassmannian is shown in Fig.~\ref{fig:subfig-grassmannian}.
\begin{figure}[htbp]
    \centering
    \includegraphics[width = 0.5\linewidth]{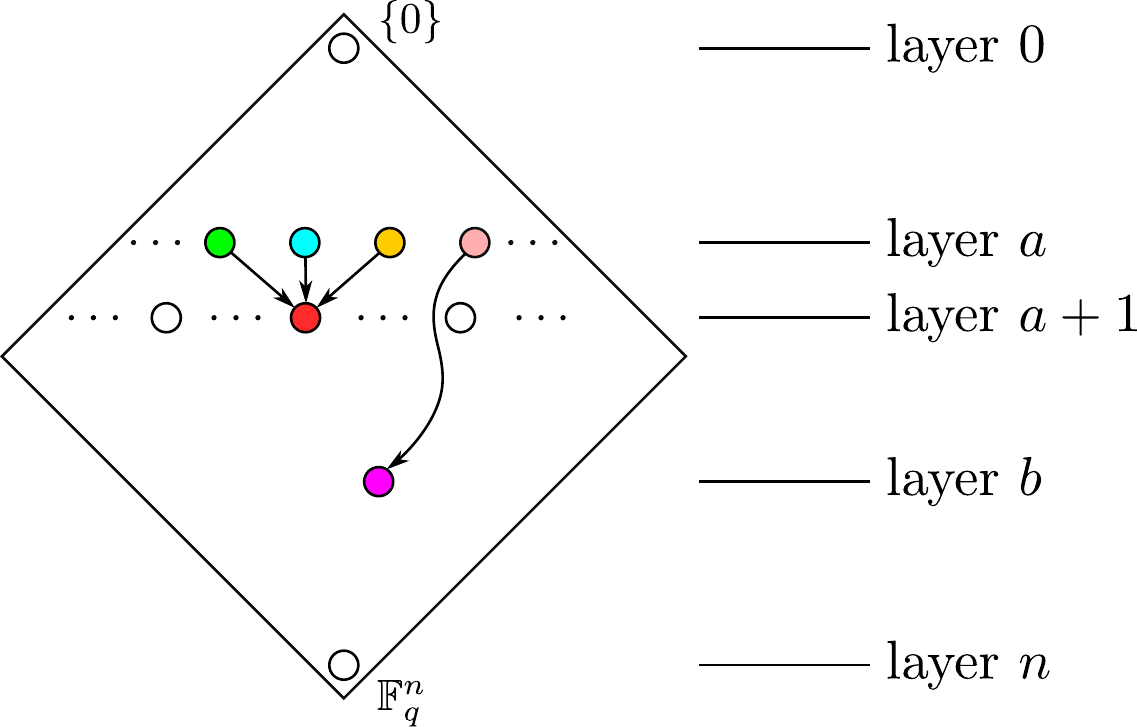}
    \caption{\it \small{A graph representation of the Grassmannian. 
    Nodes in layer $a$ are $a$-dimensional subspaces. An arrow connecting two subspaces in two adjacent layers means that the $(a+1)$-dimensional subspace in layer $a+1$ contains the $a$-dimensional subspace in layer $a$. For example, the {\color{red}red} subspace contains the {\color{brightgreen}green}, {\color{aqua}blue} and {\color{carrotorange}yellow} subspaces. 
     We say two subspaces in layers $a$ (the subspace in {\color{babypink}light pink}) and $b$ (the subspace in {\color{brightpink}dark pink}), $b>a+1$, are connected if there exists a path (series of arrows) connecting the subspace in layer $a$ to the subspace in layer $b$. }}
    \label{fig:subfig-grassmannian}
\end{figure}

It is known~\cite{metric} that the Gaussian coefficient defined as
\begin{align*}
    \binom{n}{k}_q & \triangleq \prod_{i = 0}^{k-1} \frac{q^n-q^i}{q^k - q^i},
\end{align*}
measures the cardinality of the Grassmannian $\mathcal{G}_q (n,k)$. The value of $\binom{n}{k}_q$ is bounded between $q^{kn-k^2}$ and $4q^{kn-k^2}$ \cite[Lemma~4]{KK08}. Throughout the paper, we will use those values to bound the Gaussian coefficient from below and from above, respectively. 

A subspace code is a non-empty collection of subspaces of $\mathbb{F}_q^n$. Hence
a subspace codeword is a subspace in the collection.
However, in the network communication model outlined in the previous section, codewords correspond to $\capacity 
\times n$ matrices. To be able to use the nice machinery of subspace codes, we identify any given subspace of dimension $k$ with the unique Reduced Row Echelon Form (RREF) $k \times n$ matrix  $X$ whose row-space $V(X)$ equaling the given subspace. 
Subspace codes such that each subspace in the code is of the same dimension is called a constant-dimension code. The distance function $d(.,.)$ we use %
is the {\it injection distance} between subspaces \cite{metric}, where the distance between any two subspaces $V$ and $V'$ is expressed as
\begin{equation}
    d(V,V') = \max\{\dim(V),\dim(V')\} - \dim(V\cap V'). \label{eqn:distance-metric}
\end{equation}
It is shown in~\cite{metric} that this definition results in a metric.
The injection distance is depicted in view of the Grassmannian in Fig.~\ref{fig:injection-dist}. 
\begin{figure}[htbp]
    \centering
    \includegraphics[width = 0.7\linewidth]{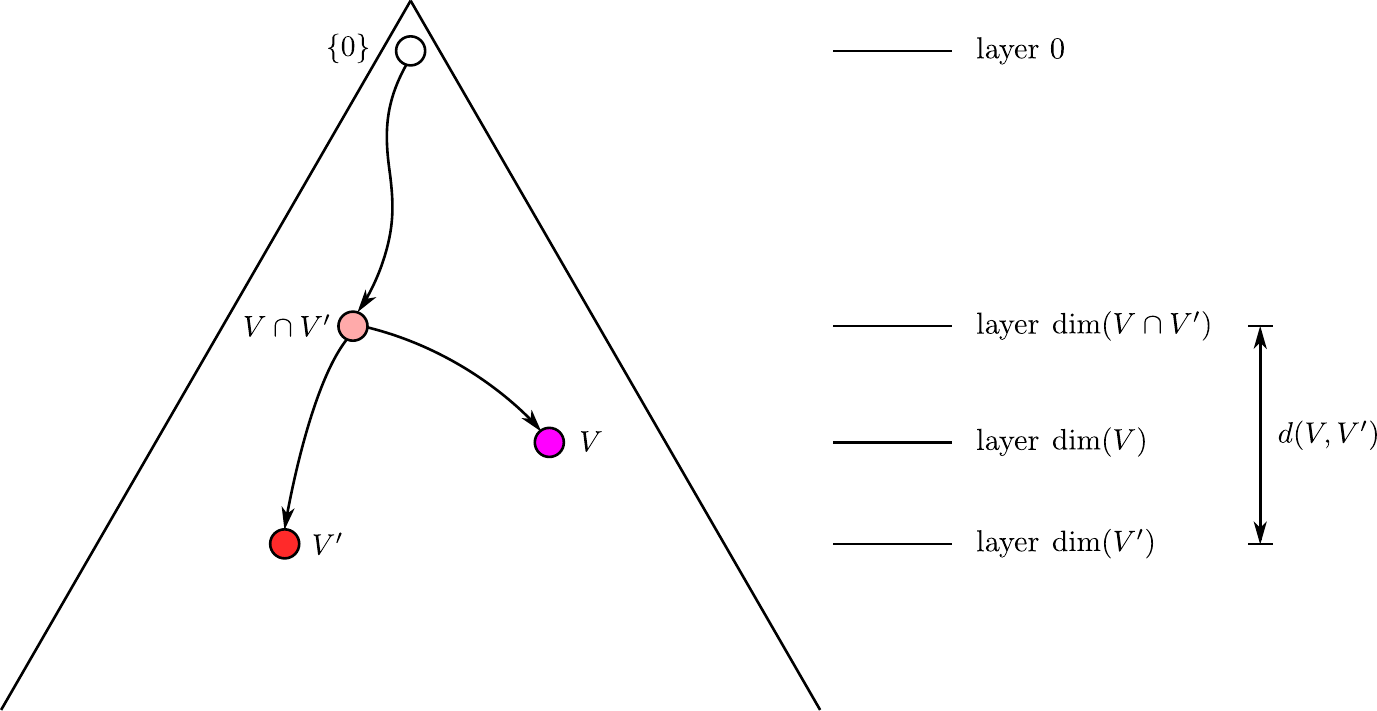}
    \caption{\it \small{Injection distance in the context of the Grassmannian graph.}}
    \label{fig:injection-dist}
\end{figure}

\subsection{Communication Scheme}

We now describe the specific encoding and decoding strategies in our scheme, and James' possible eavesdropping/jamming actions, all in the context of subspace codes over Grassmannians.

\subsubsection{Random code construction/Encoder}
We construct the codebook $\mathcal{C}$ by sampling $q^{n\rate}$ codewords (subspaces) uniformly at random from the Grassmannian $\mathcal{G}_q (n,\capacity)$. Given a message $m$ and the corresponding codeword/subspace, Alice's encoder then merely transmits the RREF matrix $X$ with row-space equaling the given subspace.


\subsubsection{Decoder} \label{decoder}
Bob uses a ``brute-force'' injection-distance decoder. The decoder measures the injection distance given in Eqn.~\eqref{eqn:distance-metric} between the received subspace $V(Y)$ and each codeword $V(X)$ in the codebook. If there is a unique codeword $\hat{X}$ such that $d(Y,\hat{X}) \leq z_w$, then the decoder outputs $\hat{X}$ as the transmitted codeword. Otherwise, the decoder outputs an error.
 
\begin{remark}
Using such an injection-distance decoder is not in general optimal for general codes (beyond subspace codes), since many different matrices may have the same row-space. Indeed, in some contexts, ignoring such degeneracy can result in loss of useful information. For instance, in~\cite{yao2012passive} a scheme that does {\it  not} collapse multiple matrices into a single subspace 
allows in some scenarios one to estimate the topology  of a given network and adversarial location. However, for our purposes in this work, where we are focused solely on the problem of characterizing the information-theoretically optimal rate of communication over networks containing myopic adversaries, exploiting the non-degeneracy of general codes (rather than subspace codes) does not asymptotically improve the throughput. And on the flip side, as has been noted in the literature in the past (for instance~\cite{KK08}), subspace codes have the pleasing property that they allow one to ignore the role that specific network topologies/linear network coding operations play in how information is transformed in the network, enabling significantly cleaner and easier analysis.
\end{remark}
 
\subsubsection{Adversarial Action}

\begin{figure}[htbp]
    \centering
    \includegraphics[width = \linewidth]{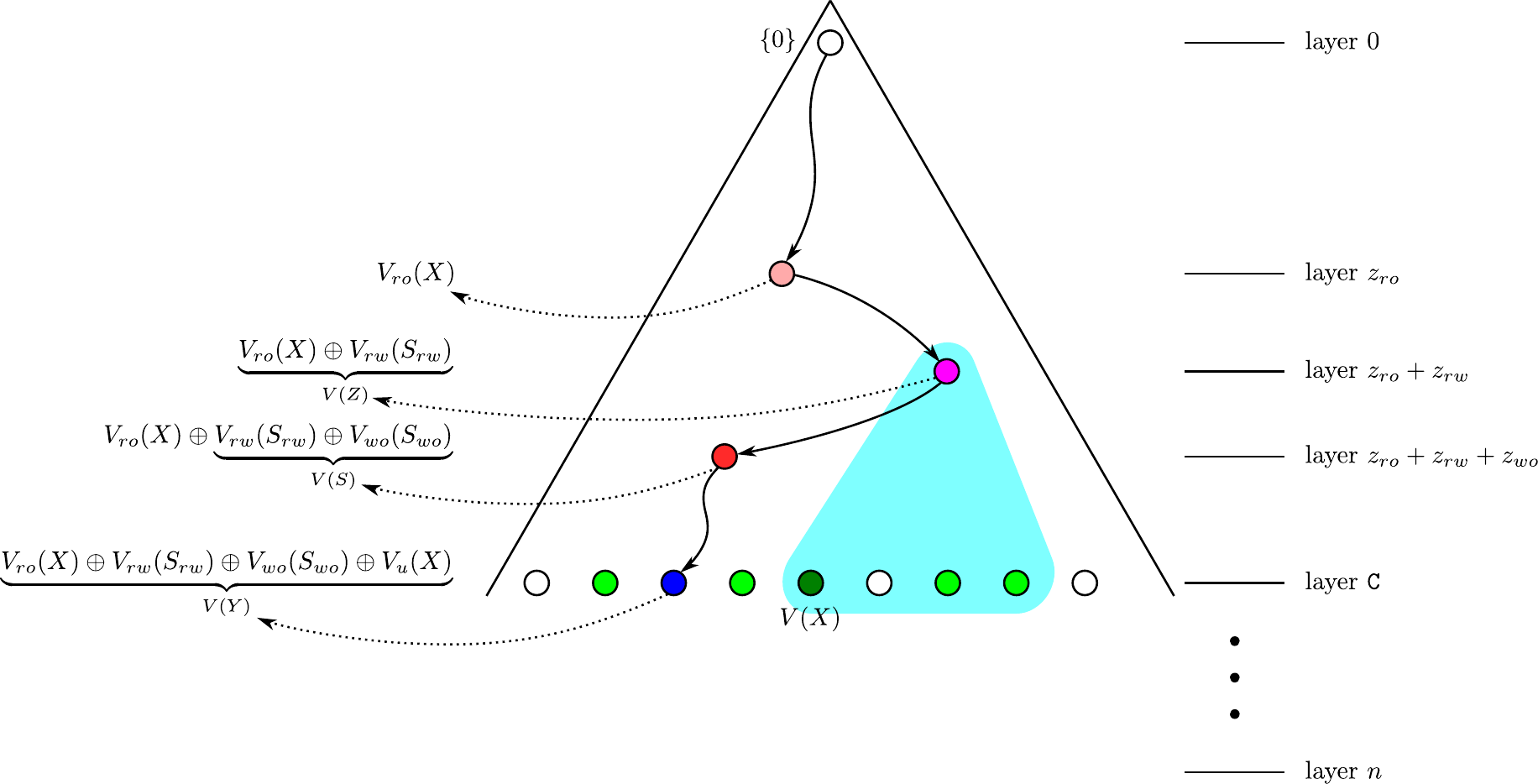}
    \caption{\it \small{James' attack in the context of the Grassmannian graph. Each layer here is a collection of subspaces over $\mathbb{F}_q$ of the same dimension. 
    The codebook consists of $ q^{n\rate} $ subspaces (the subspaces in {\color{brightgreen}light green} and {\color{cadmiumgreen}dark green}) on layer $ \capacity $. 
    The subspace in {\color{cadmiumgreen}dark green} denotes Alice's codeword $V(X)$. James has the power to first observe a subspace of Alice's codeword sitting on layer $z_{ro}+z_{rw}$ (corresponding to the subspaces in {\color{babypink} light pink} and {\color{brightpink} dark pink}). 
    Then, the number of codewords in layer $\capacity$ (the green subspaces in the {\color{aqua}light blue} shaded region) connected to James' observation $V(Z)$ (the subspace in {\color{brightpink}dark pink}) is exponential by Lemma \ref{James-uncertainty}. James can choose to overwrite at most $z_{rw}$ number of observed links (corresponding to the subspace in {\color{brightpink}dark pink}) and another $z_{wo}$ links (corresponding to the {\color{red}red} subspace). Overall, James can control the subspace $V(Y)$ (the subspace in {\color{smalt}dark blue}) received by Bob up to layer $z_{ro}+z_{rw}+z_{wo}$ by Lemma \ref{James' effect}. Our communication scheme works under the assumption given by Eqn.~\eqref{eq:weak-adversary}.}}
    \label{fig:subfig-subspaces-on-gr}
\end{figure}

James observes a $z_r$-dimensional subspace $V(Z)=V(T_{AJ}X)$ of $V(X)$ -- for notational convenience we henceforth denote this $V(Z)$ as $V_r(X)$. Based on this observation, he designs a subspace $V(S)$ of dimension not exceeding his jamming power $z_w = z_{wo}+z_{rw}$ to inject in the network. Let $V(S) = V_{wo}(S_{wo}) \oplus V_{rw}(S_{rw})$, where $V_{wo}(S_{wo})$ and $V_{rw}(S_{rw})$ respectively denote the subspaces inserted on the $z_{wo}$ write-only links and the $z_{rw}$ read-write links controlled by James. We represent the subspace $V(Y)$ received by Bob as $V_{ro}(X) \oplus V_{u}(X) \oplus V_{wo}(S_{wo}) \oplus V_{rw}(S_{rw}) $. Here $V_u(X)$ corresponds to the subspace in the direct-sum decomposition of the transmitted codeword $V(X)$ that  is neither seen nor overwritten by James.
James' adversarial action is depicted in Fig.~\ref{fig:subfig-subspaces-on-gr} in view of (the graph representation of) the Grassmannian. 

\subsection{Error Event}
We define the following to be the error event.
\begin{definition}
Consider a transmitted codeword $V(X)$ and a received subspace $V(Y)$. We say that an error happens if for some jamming action $V(S)$, there exits a codeword $V(\hat{X})\in \mathcal{C}$ such that $V(\hat{X}) \neq V(X)$ and $d(V(\hat{X}), V(Y)) \leq z_w$ {where $ V(Y) $ results from $ V(X) $ and $ V(S) $}.
\end{definition}

The probability of decoding error is the probability of finding a suitable subspace $V(S)$ 
{among all feasible jamming matrices 
} such that there exists a codeword $\hat{X}\neq X$ that satisfies $$d(V(\hat{X}), V(Y)) = d(V(\hat{X}), V_{ro}(X) \oplus V_{u}(X) \oplus V_{wo}(S_{wo}) \oplus V_{rw}(S_{rw}))\leq z_w.$$


\subsection{Main results}
\label{sec:main-results}

With the preliminaries out of the way, our main result is summarized in the following theorem.
\begin{theorem}
The network error-correction capacity $\rate^{\ast}$ of a network ${\cal N}$ with min-cut $\capacity$ and adversarial power-tuple $(z_{ro},z_{wo},z_{rw})$ equals
\begin{equation}
    \rate^{\ast} = \left \{
    \begin{array}{ll}
\cut-z_w & \mbox{ if } \cut>z_{ro}+2z_w,\\
(\cut-2z_w)^+ & \mbox{ otherwise.}
    \end{array}
    \right .
\end{equation}
\label{main theorem}
\end{theorem}
The rate converse of $\cut - z_w$ in the weak adversary regime follows directly from information-theoretic arguments; the rate converse of $(\cut - 2z_w)^+$ in the strong adversary regime relies on a myopic symmetrization attack that James can carry out. The proofs of these converses are presented in Sec.~\ref{sec:converse}. The achievability proof is more involved -- a sketch is presented in 
Sec.~\ref{sec:sketch}
and a detailed proof is presented in~\ref{sec:detail}. 

Further, the following corollary can be derived directly via a standard coset coding argument. A detailed explanation is given in Sec.~\ref{sec:coset_code}.

\begin{corollary}
The secure network error-correction capacity $\rate^{\ast}_{\text{sec}}$ of a network ${\cal N}$ with min-cut $\capacity$ and adversarial power-tuple $(z_{ro},z_{wo},z_{rw})$ equals
\begin{equation}
    \rate^{\ast}_{\text{sec}} = \left \{
    \begin{array}{ll}
\cut-z_w - z_r & \mbox{ if } \capacity>z_{ro}+2z_w,\\
0 & \mbox{ otherwise.}
    \end{array}
    \right .
\end{equation}
\label{cor:secure_capacity}
\end{corollary}



\section{Analysis}\label{sec: proof}

We first prove Theorem~\ref{main theorem} in two parts. In the first part we prove the achievability of the rate $\cut-z_w-\epsilon$ using the proposed subspace codes. Then, we use the result of \cite{converse} to show that the optimal rate is indeed $\cut-z_w$ if $\cut> z_{ro}+2z_w$ and is $\cut-2z_w$ otherwise. Then, Corollary~\ref{cor:secure_capacity} follows immediately by coupling our subspace code with a coset code~\cite{el2012secure} for the achievability part. We explain the idea of coset coding and how it is coupled with our subspace code in Section~\ref{sec:coset_code}. The converse follows from~\cite{converse}.

We start with a sketch of the achievability proof of Theorem~\ref{main theorem} to provide  intuition.

\subsection{Sketch of the achievability proof}\label{sec:sketch}

Recall that the codebook is constructed by choosing $q^{n\rate}$ subspaces uniformly at random from the Grassmannian $\mathcal{G}_q (n,\capacity)$. The decoder decodes the received $V(Y)$ as explained in Section~\ref{decoder}. We will show that with high probability, reliable communication in the presence of a weak adversary is possible using the subspace code with rate $\rate = \capacity-z_w-\epsilon$.

Consider the subspace $V(X)$ transmitted by Alice. James observes a random subspace $V_r(X)$ of $V(X)$ of dimension $z_r$. Since  $z_r + z_w < \capacity - z_{wo}$, we can show that approximately $q^{n(z_{wo}+\delta_1)}$, for some $\delta_1 > 0$, codewords are $Z$-compatible with James' observation $V_r(X)$ 
(Lemma \ref{James-uncertainty}). 
In other words, from James' perspective, exponentially many codewords could have been transmitted by Alice.

\emph{Oracle-given Set:} For the ease of analysis, we give James more power by giving him extra information about the $Z$-compatible codewords. After Alice decides on $V(X)$, an oracle reveals to James a random set of $q^{n\epsilon_1}$ $Z$-compatible codewords, which includes the correct one. This set is referred to as the oracle-given set and is denoted by $\mathcal{M}_{og}$. Note that this only increase James' power by reducing the number of $Z$-compatible codewords.

The oracle-given set is generated as follows. For each potential $V_r(X)$ that may be seen by James, about $q^{n(z_{wo}+\delta_1)}$ codewords are $Z$-compatible. The oracle partitions these compatible codewords into $q^{n(z_{wo}+\delta_1 - \epsilon_1)}$ sets, each of size $q^{n\epsilon_1}$. 
After eavesdropping on $z_r$ links, James is given one of the previously constructed sets denoted by $\cM_{og}$ that contains the true codeword. Since each codeword is generated randomly and each partition is constructed randomly, each codeword in $\cM_{og}$ can be viewed as generated uniformly at random conditioned on James' observation.

Now we show that for all jamming actions, Bob will be able to decode $V(X)$ successfully with high probability. Consider a certain subspace $V(X)$ transmitted by Alice. James can jam at most $z_w$ dimensions of this subspace. Bob receives the subspace $V(Y) = V_{ro}(X) \oplus V_{u}(X) \oplus V_{rw}(S_{rw}) \oplus V_{wo}(S_{wo})$. A decoding error happens if and only if there is another codeword $V(\hat{X})\in \mathcal{C}$ such that $d(V(\hat{X}),V(Y))\leq z_w$. In this case, we say that the transmitted codeword $V(X)$ is confused by the codeword $V(\hat{X})$. We shall use the following terminology: $V(X)$ is a \emph{confusable} codeword and $V(\hat{X})$ is a \emph{confusing} codeword.

We show that for a fixed subspace $V(S)$ that James injects, only a small fraction of codewords in $\mathcal{M}_{og}$ can be confused by a codeword $V(\hat{X})\neq V(X)$ in $\mathcal{C}$. We divide the codebook into two parts: (i) the codewords outside the oracle-given set, i.e., in $\mathcal{C} \setminus \mathcal{M}_{og}$; and (ii) the codewords inside the oracle-given set $\mathcal{M}_{og}$. Based on a careful analysis on the Principle of Deferred Decisions\cite{noga2004probabilistic}, we may assume that the codewords in $\mathcal{C} \setminus \mathcal{M}_{og}$ are independent of those in $\mathcal{M}_{og}$. Then we can bound the number of confusing codewords in the two previously discussed parts. One more step is needed to conclude the analysis. We need to show that a confusing codeword in $\cM_{og}$ can confuse a small number of codewords. Thus, we can bound the total number of confusable codewords in $\mathcal{M}_{og}$ which in turns bounds the probability of error. Note that the randomness of all the concentration analysis provided in the proof comes from the codebook construction. 

For part~(i), we bound the number of confusable codewords in $\mathcal{M}_{og}$ with confusing codewords from $\mathcal{C} \setminus \mathcal{M}_{og}$. We first use a basic list-decoding argument to show that for a given $V(Y)$, the number of confusing codewords in $\mathcal{C} \setminus \mathcal{M}_{og}$ that can confuse $V(X)$ is no more than $n^3$ with probability $1-e^{-\frac{n^3}{3}}$ over code design. Then, we analyze the best $V(\hat{X})$ that James can hope for to confuse $V(X)$. We use the randomness in $V_{u}(X)$ to show that the best confusing $V(\hat{X})\in \cC \setminus \cM_{og}$ can confuse at most $n^3$ codewords in $\cM_{og}$ with probability $1-e^{-\frac{n^3}{3}}$. Hence, in part~(i) of the codebook at most $n^3 \times n^3 = n^6$ of the $q^{n\epsilon_1}$ possible $Z$-compatible codewords are confused by  {\it any} choice of $S$ by James. 

For part~(ii), we bound the number of confusable codewords in $\mathcal{M}_{og}$ with confusing codewords from $\mathcal{M}_{og}$. In this case, we cannot use a basic list-decoding argument directly since it strongly relies on the independence between the set of confusing codewords and the set of confusable codewords. To overcome this difficulty, we randomly partition $\mathcal{M}_{og}$ into a $q^{n_{\epsilon_1}} \times q^{n_{\epsilon_1}}$ grid. We use the independence among the codewords in rows/columns of this grid to argue that only an exponentially small fraction of codewords can be confusing codewords. Therefore, we can use the same analysis we used for part~(i) to bound the number of confusable codewords by $2q^{n\epsilon_1/2}n^3$ for each confusing codeword (the $2q^{n\epsilon_1/2}$ comes from taking a union bound over the $2q^{n\epsilon_1/2}$ rows + columns in the grid). This analysis results in bounding the probability (over code design) of having a small fraction of confusable codewords in $\cM_{og}$ by $1-e^{-\frac{n^3}{3}}$.

Overall, a code is said to be ``bad'' if there exists a transmitted matrix $X$ and  jamming matrix $S$ such that more than $n^6$ codewords are confused in part~(i) of the codebook, or more than $2n^6q^{n\epsilon_1/2}$ codewords are confused in part~(ii).

The claimed arguments for the two parts will give us a super-exponentially small probability of a ``bad code'' for every fixed $V(S)$ and $V(X)$. Therefore, taking a union bound over all possible $V_r(X)$, $V(S)$ and $V(X)$, we still have a small probability of error. Hence, we argue that with high probability, only a small fraction of codewords in $\cM_{og}$ are confusable codewords. Therefore the probability of error is also small since from James' perspective codewords from $\cM_{og}$ are each uniformly likely to have been transmitted.

\subsection{Detailed proof of the achievability}
\label{sec:detail}
We prove that reliable communication can be achieved with a random subspace code with rate $\rate = \capacity- z_w -\epsilon$ for some sufficiently small $\epsilon$. 
The source of randomness in our analysis is the random generation of the codebook.

First, we argue that based on James' observation, there is at least $\frac{1}{2}q^{n(z_{wo}+\delta_1)}$ $Z$-compatible codewords with high probability.

\begin{lemma}
Consider a random subspace code $\mathcal{C}$ of rate $\rate = \capacity-z_w -\epsilon$ used to transmit a message through a network. Let $\delta_1 \geq 1-\epsilon $. For a given subspace $V_r(X)$ of $V(X)$ that a weak adversary obtains by observing a random set of $z_r$ links of the network, the following holds%
\begin{equation*}
    \Pr\left(\text{\# of $Z$-compatible codewords with a given } V_r(X)\leq \frac{1}{2}q^{n(z_{wo}+\delta_1)}\right) \leq e^{-\frac{1}{8}q^{n(z_{wo}+\delta_1)}}.
\end{equation*}
\label{James-uncertainty}
\end{lemma}
\begin{proof}
Since James observes a $z_r$-dimensional subspace, the remaining subspace is still uniformly distributed from his perspective. The cardinality of the remaining subspace is $\binom{n}{\capacity-z_r}_q$. 
Thus, the probability of a codeword being compatible with James' observation is $\frac{\binom{n}{\capacity-z_r}_q}{\binom{n}{\capacity}_q}$, which is no larger than $4q^{-nz_r +2\capacity z_r -z_r^2}$ (using the bounds on the Gaussian coefficient). We compute the expected number of $Z$-compatible codewords as%
\begin{align*}
 \mathbb{E}\left[\text{\# of }Z\text{-compatible codewords}\right]
    & = \frac{\binom{n}{\capacity-z_r}_q}{\binom{n}{\capacity}_q}q^{n\rate}\\
    & \geq 4q^{-nz_r +2\capacity z_r -z_r^2} q^{n\rate} \\
     &= 4q^{n(\capacity-z_{ro}-2z_{rw}-z_{wo}-\epsilon)+2\capacity z_r-z_r^2}\\
     & \geq q^{n(z_{wo}+\delta_1)}.
\end{align*}
The last inequality holds for $\delta_1 \geq 1-\epsilon$, since $\capacity>z_{ro}+2(z_{rw}+z_{wo})$ and $\capacity$, $z_{ro}$, $z_{rw}$ and $z_{wo} $ are integers. Then, by applying the lower tail of the Chernoff bound~\cite[Eqn.~(1.10.12)]{ben2018prob}, we can bound the number of $Z$-compatible codewords as given in the statement of the lemma.

\end{proof}

Lemma~\ref{James-uncertainty} quantifies the probability of obtaining more than $\frac{1}{2}q^{n(z_{wo}+\delta_1)}$ $Z$-compatible codewords given a fixed observation $V_r(X)$. To bound the probability of obtaining more than $\frac{1}{2}q^{n(z_{wo}+\delta_1)}$ $Z$-compatible codewords for all possible observations, we take the union bound over all $\binom{\capacity}{z_r}_q$ possible observations as follows
\begin{align*}
     \Pr\left(\text{\# of $Z$-compatible codewords with 
     \emph{any} } V_r(X)\leq \frac{1}{2}q^{n(z_{wo}+\delta_1)}\right) &\leq e^{-\frac{1}{8}q^{n(z_{wo}+\delta_1)}} \binom{\capacity}{z_r}_q\\
     & \leq 4q^{nz_r-z_r^2}e^{-\frac{1}{8}q^{n(z_{wo}+\delta_1)}}.
\end{align*}
\noindent We conclude that the probability of James having less than $\frac{1}{2}q^{n(z_{wo}+\delta_1)}$ $Z$-compatible codewords is exponentially small in $n$.


Afterwards, we can reveal the oracle-given set $\mathcal{M}_{og}$ to James and analyse the probability of error in the two following cases.

\subsubsection{Type-I Error}

In this case, we consider the confusing codeword $\hat{X}$, i.e. the codeword that may confuse Bob with the true codeword, is in $\mathcal{C} \setminus \mathcal{M}_{og}$. Recall that the set $\mathcal{C} \setminus \mathcal{M}_{og}$ is considered independent from the $\mathcal{M}_{og}$.

We use the following lemma to argue that with high probability there are at most $n^3$ confusing codewords in $\mathcal{C} \setminus \mathcal{M}_{og}$.

\begin{lemma}
 For any $V_r(X)$ observed by a weak adversary through the network, let $V(X_i)\in \cM_{og}$ be a $Z$-compatible codeword and let $V(S)$ be the adversary's jamming action. For $V(Y_i) = V_{ro}(X) \oplus V_{u}(X_i) \oplus V(S)$, define the decoding region of $V(Y_i)$ as $\cD(Y_i) \triangleq \{V(X)\;|\; \dim(V(X))= \cut,\ d(V(X),V(Y_i))\leq z_w\}$. Then, based on the random generation of the codebook $\mathcal{C}$, we can write
\begin{equation*}
    \Pr\left(\exists V_r(X), V(S),\text{ s.t. } \left| \bigcup_{V(X_i)\in \cM_{og}} \cD(Y_i) \cap (\cC \setminus \cM_{og})\right|\geq n^3\right) \leq e^{-\frac{\alpha}{3} n^3}.
\end{equation*} 
In other words, for any $V_r(X)$ observed by James and any $V(S)$ that James inject into the network, the probability that more than $n^3$ codewords in $\mathcal{C} \setminus \mathcal{M}_{og}$ fall into the union of the decoding regions $\cD(Y_i)$ of codewords $V(X_i) \in \mathcal{M}_{og}$ is bounded from above by $e^{-\frac{\alpha}{3} n^3}$.

\label{Type 1 list decoding}
\end{lemma}
\begin{proof}
We start by bounding the cardinality of the decoding region $\cD(Y)$ for a given received codeword $V(Y)$. Since our decoding strategy is to decode to a codeword within distance $z_w$ from $V(Y)$, the cardinality $|\cD(Y)|$ is bounded as
\begin{align}
    |\cD(Y)|& = \sum_{i=1}^{z_w}\binom{\capacity}{i}_q \binom{n}{i}_q \nonumber\\
    &\leq z_w \binom{\capacity}{z_w}_q \binom{n}{z_w}_q \nonumber \\
     & \leq  16z_w q^{\capacity z_w - z_w^2} q^{nz_w - z_w^2}.\label{eq:cardinality_dec_reg}
\end{align}

Given an observation $V_r(X)$, its corresponding $\cM_{og}$ and a jamming action $V(S)$, let $\Lambda \triangleq \bigcup_{V(X_i)\in \cM_{og}}\cD(Y_i)$ be the union of the decoding regions corresponding to all codewords in $ \cM_{og} $. Then the cardinality of $\Lambda$ is at most $16q^{n\epsilon_1} z_w q^{\capacity z_w - z_w^2} q^{nz_w - z_w^2}$. The probability that a codeword in $\mathcal{C} \setminus \mathcal{M}_{og}$ falls in $\Lambda$, i.e., is confusing, is given by
\begin{align}
    \Pr\left( V(X) \in \Lambda \cap \mathcal{C} \setminus \mathcal{M}_{og}\right) & =\frac{16q^{n\epsilon_1} z_w q^{\capacity z_w - z_w^2} q^{nz_w - z_w^2}}{q^{n\capacity-\capacity^2}}\\
    & = 16z_w q^{-n(\capacity-\epsilon_1 -z_w -\frac{\capacity z_w - 2z_w^2}{n})}\label{eq:decoding_region}
\end{align}
Then the expected number of codewords from $\mathcal{C} \setminus \mathcal{M}_{og}$  that fall in $\Lambda$ is
\begin{align*}
    \mathbb{E}\left[|\Lambda\cap \cC\setminus\cM_{og}|\right] 
    & = 16z_w q^{-n(\capacity-\epsilon_1 -z_w -\frac{\capacity z_w - 2z_w^2}{n})} q^{n(\capacity-z_w-\epsilon-\epsilon_1)} \\
     &= 16 q^{-n(\epsilon - 2\epsilon_1 - \frac{\capacity z_w - 2z_w^2-\log_q(z_w)}{n})}\\
     & = q^{-n\epsilon_2}
\end{align*}
We can adjust $\epsilon$ and $\epsilon_1$ such that $\epsilon > 2\epsilon_1 + \frac{\capacity z_w - 2z_w^2-\log_q(z_w)}{n}$ to make sure that $\epsilon_2$ is a positive number. Therefore, the expected number of confusing codewords in $\mathcal{C} \setminus \mathcal{M}_{og}$ is exponentially small. Then, we can apply the upper tail of Chernoff bound~\cite[Eqn.~ (1.10.4)]{ben2018prob} to argue that the probability of having more than $n^3$ confusing codewords in $\mathcal{C} \setminus \mathcal{M}_{og}$ is bounded from above by $e^{-\frac{n^3}{3}}$. With this super-exponentially small probability, we can take the union bound over the number of $V_r(X)$ and $V(S)$, which are all of size exponential in $n$, to argue that: 
\begin{equation*}
    \Pr\left(\exists V_r(X), V(S) ,\text{ s.t. } |(\cC\setminus\cM_{og})\cap\Lambda| \geq n^3\right) \leq e^{-\frac{\alpha}{3} n^3}
\end{equation*} 
for some $\alpha>0$ to be obtained after taking the union bound.
\end{proof}

Next, we need to show that each confusing codeword in $\mathcal{C} \setminus \mathcal{M}_{og}$ can confuse at most $n^3$ codewords in $\mathcal{M}_{og}$ with high probability. This means that for a confusing codeword $V(\hat{X})$ in $\mathcal{C} \setminus \mathcal{M}_{og}$, the number of codewords $V(X_i)\in\cM_{og}$ such that $V(\hat{X})$ falls in the decoding region of $V(Y_i) = V_{ro}(X) \oplus V_{u}(X_i) \oplus V(S)$ is no more than $n^3$.

As an intermediate step, we show that James' best attack can reveal to him a $(z_{ro}+z_{rw}+z_{wo})$-dimensional subspace of the codeword $V(Y)$ received by Bob. 

\begin{lemma}
James can either observe or control at most a $(z_{ro}+z_{rw}+z_{wo})$-dimensional subspace of $V(Y)$ that Bob receives.
\label{James' effect}
\end{lemma}

\begin{proof}
The key idea here is to argue that the $z_{rw}$ links can only reveal to James a subspace of dimension $z_{rw}$ of the received $V(Y)$. In other words, if James decides to eavesdrop and overwrite one of the $z_{rw}$ links, then the total dimension revealed to James about $V(Y)$ using this link is one (either the subspace he reads, or the subspace he inserts). To see this, recall that James eavesdrops $z_{ro}$ links and can blindly jam $z_{wo}$ links. This reveals to James a subspace of dimension $z_{ro}+z_{wo}$ about $V(Y)$. Assume that James observes and overwrites one link of the $z_{rw}$ links and only eavesdrops on the others. If our argument does not hold, then this action reveals to James a subspace of dimension $z_{rw}+1$ about $V(Y)$, i.e., those $z_{rw}$ links contribute to $V(Y)$ with a subspace of dimension $z_{rw}+1$. This is a contradiction of the network code for the following reason. For James to learn a $z_{rw}$-dimensional subspace about $V(Y)$ from those links, he will observe a $z_{rw} \times \capacity$ transfer matrix  multiplying a dimension $\capacity$ code. And this $z_{rw} \times \capacity$ transfer matrix is part of a full-rank $\capacity \times \capacity $ transfer matrix. Therefore, those links can contribute in a subspace of dimension at most $z_{rw}$ to $V(Y)$. Which means that every subspace added by James on those links wipes out the subspace he deleted from $V(Y)$ and his total knowledge from those links about $V(Y)$ is a $z_{rw}$-dimensional subspace. 
Thus, the dimension of subspace that James can make sure to appear at Bob's side is $z_{ro}+z_{rw}+z_{wo}$.
\end{proof}

According to Lemma~\ref{James' effect}, the received $V(Y)$ can be expressed as $V(Y) = V_{ro}(X) \oplus V_{u}(X) \oplus V_{rw}(S_{rw}) \oplus V_{wo}(S_{wo})$ where $V_{u}(X)$ is distributed uniformly at random from James' perspective. Since each codeword is chosen uniformly at random, it is as if each $V_u(X)$ is also chosen uniformly random. We will use the randomness in $V_{u}(X)$ to argue that a confusing codeword in $\mathcal{C} \setminus \mathcal{M}_{og}$ can confuse at most $n^3$ codewords in $\mathcal{M}_{og}$ with high probability.

\begin{lemma}
For any $V_r(X)$ observed by a weak adversary through the network and for any $V(S)$ that the adversary injects into the network, let $V(\hat{X})$ be a confusing codeword in $\mathcal{C} \setminus \mathcal{M}_{og}$. Define $\cD_1  \triangleq \{V(X_i)\in \mathcal{M}_{og} \;|\; d(V_{ro}(X) \oplus V_{u}(X_i) \oplus V(S), V(\hat{X}))\leq z_w\}$ as the set of codewords in $\cM_{og}$ confusable by $V(\hat{X})$. Based on the randomness in $V_u(X)$, we can write
\begin{equation*}
    \Pr\left(\exists V_r(X), V(S),\text{ s.t. } | \cD_1 |\geq n^3\right) \leq e^{-\frac{\alpha_1}{3} n^3}.
\end{equation*} 
In other words, the probability that a confusing codeword in $\mathcal{C} \setminus \mathcal{M}_{og}$ confuses more than $n^3$ codewords in $\mathcal{M}_{og}$ is bounded from above by $e^{-\frac{\alpha_1}{3} n^3}$.

\label{Type 1 reverse list decoding}
\end{lemma}
\begin{proof}
The more powerful omniscient adversary that knows the codebook and the transmitted message operates by carefully inserting errors to push the received codeword to the closest codeword (different than the transmitted one) in the codebook. To emulate such a powerful adversary, the best strategy for James (a weak adversary) is to fully leverage his knowledge and push $V(Y)=V_{ro}(X)\oplus V_u(X)\oplus V(S)$ towards $V(\hat{X})$ that satisfies $V(\hat{X}) = V_{ro}(X) \oplus V_{u}(\hat{X}) \oplus V(S)$. This means that James chooses a codeword compatible with his observation, deemed to be close to the transmitted codeword, and tries to push $V(Y)$ towards that codeword. 
Otherwise, James will not be using his power efficiently. 
We do the analysis for a given $V_r(X)$ and a given $V(S)$. We then take a union bound over all possible $V_r(X)$ and $V(S)$.

A codeword $V(\hat{X})$ confuses $V(X_i)\in \cM_{og}$ if the following holds $$d(V_{ro}(\hat{X})\oplus V_u(\hat{X})\oplus V(S), V_{ro}(X_i)\oplus V_u(X_i)\oplus V(S))\leq z_w.$$

Notice that James has full control over $V(S)$ and that all codewords $V(X_i)\in \cM_{og}$ satisfy $V_{ro}(X_i) = V_{ro}({X})$. To maximize the number of confusable codewords, James must chose $V_{ro}(\hat{X}) = V_{ro}(X)$. The only uncertainty that remains from James' perspective is in $V_u(X)$ which remains uniformly distributed. Therefore, all we have to count is the number of codewords in $\cM_{og}$ that satisfy $d(V_u(\hat{X}), V_u(X_i))\leq z_w$. By definition of the injection distance given in Eqn.~\eqref{eqn:distance-metric}, this implies that $\dim(V_u(\hat{X})\cap V_u(X_i))\geq \capacity-z_w$. 

Recall that $\dim(V_u(X_i)) = \dim(V_u(\hat{X})) = \capacity - z_{ro} - z_w$.
Therefore, for a fixed $V_u(\hat{X})$, the number of $V_u({X_i})$ that have an intersection of dimension at least $\capacity-2z_w-z_{ro}$ with $V_u(\hat{X})$ is bounded by
\begin{align*}
    \sum_{i=1}^{z_w} \binom{\capacity-z_{ro}-z_{w}}{\capacity-z_{ro}-z_{w}-i}_q \binom{n}{i}_q 
    &  \leq z_w \binom{\capacity-z_{ro}-z_{w}}{\frac{1}{2}(\capacity-z_{ro}-z_{w})}_q \binom{n}{z_w}_q\\
    &  \leq 4z_w q^{(\capacity-z_{ro}-z_w)^2\frac{1}{4}} q^{nz_w - z_w^2}
\end{align*}
Thus, the probability that $V(\hat{X})$ confuses $V_{ro}(X) \oplus V_{u}(X_i) \oplus V(S)$ is bounded as follows
\begin{align*}
    \Pr(V(\hat{X}) \text{ confuses }V(X_i)) & \leq \frac{4z_w q^{(\capacity-z_{ro}-z_w)^2\frac{1}{4}} q^{nz_w - z_w^2}}{\binom{n}{\capacity-z_{ro}-z{w}}_q}\\
    &\leq 4z_w q^{-n(\capacity-z_{ro}-2z_w-\frac{(\capacity-z_{ro}-z_w)^2\frac{1}{4}-z_w^2+(\capacity-z_{ro}-z_w)^2}{n})} \\ 
    & \leq 4z_w q^{-n}.
\end{align*}
The last inequality holds since $\capacity>z_{ro}+2z_w$ and all the numbers in the exponent of $q$ are integers. Thus, the expected number of codewords in $\mathcal{M}_{og}$ that a confusing codeword in $\cC \setminus \mathcal{M}_{og}$ can confuse is $4z_w q^{-n} q^{n\epsilon_1} = q^{-n\epsilon_3}$. Then, by the upper tail of the Chernoff bound~\cite[Eqn.~(1.10.4)]{ben2018prob} we can argue that the probability that $V(\hat{X})$ can confuse more than $n^3$ codewords in $\mathcal{M}_{og}$ is no more than $e^{-\frac{n^3}{3}}$.
Then we take union bound over $V_r(X)$ and $V(S)$, which are both exponentially in $n$, to show that the lemma holds with some $\alpha_1 > 0$.
\end{proof}

Based on Lemma~\ref{Type 1 list decoding} and Lemma~\ref{Type 1 reverse list decoding}, we can argue that with high probability at most $n^6$ codewords in $\mathcal{M}_{og}$ can be confused by a confusing codeword in $\mathcal{C} \setminus \mathcal{M}_{og}$. 

\subsubsection{Type-II Error}
In this case, we consider the confusing codeword $V(\hat{X})$ to be in $\mathcal{M}_{og}$. Recall that the size of $\mathcal{M}_{og}$ is $q^{n\epsilon_1}$. We pick each element of $\mathcal{M}_{og}$ uniformly at random from the set of all compatible codewords. Recall that the number of possible compatible codewords is $\binom{n}{\capacity-z_r}_q$.

Then we arrange the codewords in $\mathcal{M}_{og}$ in the following way:
\begin{enumerate*}[label={\textit{\roman*)}}] \item Initialize a $q^{\frac{n\epsilon_1}{2}} \times q^{\frac{n\epsilon_1}{2}}$ grid; \item Arrange each codeword of $\mathcal{M}_{og}$ randomly into the grid.\end{enumerate*} We pick any row or column from the grid and refer to it as the mini-oracle-given set $\mathcal{M}_{mi}$. In this way, we divide the codewords in $\mathcal{M}_{og}$ into two parts: $\mathcal{M}_{og} \setminus \mathcal{M}_{mi}$ and $\mathcal{M}_{mi}$. Here the set $\mathcal{M}_{og} \setminus \mathcal{M}_{mi}$ is considered independent from $\mathcal{M}_{mi}$ in the unseen dimension $\capacity-z_r$ subspace.
Now we consider a codeword $V(\hat{X})$ in $\mathcal{M}_{og} \setminus \mathcal{M}_{mi}$ that may confuse a codeword in $\mathcal{M}_{mi}$.

\begin{lemma}
For any $V_r(X)$ observed by James, let $S_{rw}$ and $S_{wo}$ be his jamming action and recall the decoding region defined as $\cD(Y_i) = \{V(X)\;|\; \dim(V(X))= \cut,\ d(V(X),V(Y_i))\leq z_w\}$ where $V(Y_i) = V_{ro}(X) \oplus V_{u}(X_i) \oplus V_{rw}(S_{rw}) \oplus V_{wo}(S_{wo})$. Then, based on the randomness from the subspace $V_u(X)$ of $V(X)$ that is not observed by James, we can write
\begin{equation*}
    \Pr\left(\exists V_r(X), V(S),\text{ s.t. } \left| \bigcup_{V(X_i)\in \cM_{mi}} \cD(Y_i) \cap (\cM_{og} \setminus \cM_{mi})\right|\geq n^3\right) \leq e^{-\frac{\alpha_2}{3} n^3}.
\end{equation*} 
In other words, for any $V_r(X)$ observed by James and any $V(S)$ that James inject into the network, the probability that more than $n^3$ codewords in $\mathcal{M}_{og} \setminus \mathcal{M}_{mi}$ fall into the union of the decoding regions $\cD(Y_i)$ of codewords $X_i\in \mathcal{M}_{mi}$ is bounded from above by $e^{-\frac{\alpha_2}{3} n^3}$. 

\label{Type 2 list decoding}
\end{lemma}
\begin{proof}
Recall from Eqn.~\eqref{eq:cardinality_dec_reg} that the cardinality of the decoding region of $V(Y)$ is no more than $ 16z_w q^{\capacity z_w - z_w^2} q^{nz_w - z_w^2}$. Given an observation $V_r(X)$, the set $\cM_{mi}$ and a jamming action $V(S)$, let $\Lambda_{mi} \triangleq \bigcup_{V(X_i)\in \cM_{mi}}\cD(Y_i)$ be the union of the decoding regions corresponding to all codewords in $ \cM_{mi} $. Then, the cardinality of $\Lambda_{mi}$ is at most $16z_w q^{\capacity z_w-z_w^2}q^{nz_w-z_w^2}q^{n\frac{\epsilon_1}{2}}$. The probability that a codeword in $\mathcal{M}_{og} \setminus \mathcal{M}_{mi}$ falls in $\Lambda_{mi}$, i.e., is confusing, is given by 
\begin{align*}
    \frac{16z_w q^{\capacity z_w-z_w^2}q^{nz_w-z_w^2}q^{n\frac{\epsilon_1}{2}}}{q^{n(\capacity-z_r)-(\capacity-z_r)^2}} 
    &= 16z_w q^{-n(\capacity-z_r-z_w-\frac{\epsilon_1}{2}+\frac{2z_w^2-\capacity z_w}{n})}\\
    &= q^{-n\epsilon_4}
\end{align*}
where the denominator is the size of the subspace that is not observed by James and $\epsilon_4$ is some positive number. Hence, the expected number of confusing codewords is $q^{-n\epsilon_4} q^{n\epsilon_1} = q^{-n\epsilon_5}$ for some positive $\epsilon_5 = \epsilon_4 - \epsilon_1$. By applying the upper tail of the Chernoff bound~\cite[Eqn.~(1.10.4)]{ben2018prob}, we conclude that the probability of having more than $n^3$ confusing codewords in $\mathcal{M}_{og} \setminus \mathcal{M}_{mi}$ is bounded from above by $e^{-\frac{n^3}{3}}$. 

Taking the union bound over the size of $V_r(X)$ and $V(S)$, we can argue that with probability at most $e^{-\frac{\alpha_2}{3}n^3}$ for some coefficient $\alpha_2 > 0$, there are more than $n^3$ confusing codewords in $\mathcal{M}_{og} \setminus \mathcal{M}_{mi}$.
\end{proof}

Next we need to show that each confusing codeword in $\mathcal{M}_{og} \setminus \mathcal{M}_{mi}$ can confuse at most $n^3$ codewords in $\mathcal{M}_{mi}$. The argument is similar to the one made in Lemma~\ref{Type 1 reverse list decoding}.

\begin{lemma}
For any $V_r(X)$ observed by a weak adversary through the network and for any $V(S)$ that the adversary injects into the network, let $V(\hat{X})$ be a confusing codeword in $\mathcal{M}_{og} \setminus \mathcal{M}_{mi}$. Define $\cD_2  \triangleq \{V(X_i)\in \mathcal{M}_{mi} \;|\; d(V_{ro}(X) \oplus V_{u}(X_i) \oplus V(S), V(\hat{X}))\leq z_w\}$ as the set of codewords in $\cM_{mi}$ confusable by $V(\hat{X})$.
Based on the randomness in $V_u(X)$
we can write 
\begin{equation*}
    \Pr\left(\exists V_r(X), V(S),\text{ s.t. } | \cD_2 |\geq n^3\right) \leq e^{-\frac{\alpha_3}{3} n^3}.
\end{equation*} 
In other words, the probability that a confusing codeword in $\mathcal{M}_{og} \setminus \mathcal{M}_{mi}$ confuses more than $n^3$ codewords in $\mathcal{M}_{mi}$ is bounded from above by $e^{-\frac{\alpha_3}{3} n^3}$.
\label{Type 2 reverse list decoding}
\end{lemma}
\begin{proof}
The proof is similar to the proof of Lemma~\ref{Type 1 reverse list decoding}. First by Lemma~\ref{James' effect}, despite that James manages to eavesdrop on $z_r$ links and overwrite $z_w$ links, he can control a subspace of dimension at most $z_{ro} + z_{rw} + z_{wo}$ in $V(Y)$. Thus, the same analysis made in
Lemma~\ref{Type 1 reverse list decoding} holds and $V(\hat{X})$ can confuse $V_{ro}(X) \oplus V_{u}(X_i) \oplus V(S)$ with probability at most $4z_w q^{-n}$. Then, the expected number of confusable codewords in $\mathcal{M}_{mi}$ is $4z_w q^{-n} q^{\frac{n\epsilon_1}{2}} = q^{-n\epsilon_6}$. By the Chernoff bound~\cite[Eqn.~(1.10.4)]{ben2018prob} we can argue that with probability at most $e^{-\frac{n^3}{3}}$, there are more than $n^3$ codewords in $\mathcal{M}_{mi}$ that are confusable with this $V(\hat{X})$.

Taking the union bound over the size of $V_r(X)$ and $V(S)$, which are both exponentially in $n$, we argue that the lemma holds for some $\alpha_3 > 0$.
\end{proof}

Based on the Lemma \ref{Type 2 list decoding} and Lemma~\ref{Type 2 reverse list decoding}, we can argue that for one row or one column to be the $\mathcal{M}_{mi}$, and for any $V_r(X)$ and $V(S)$, the probability having more than $n^6$ confusable codewords in the $\mathcal{M}_{mi}$ is bounded from above by $e^{-\frac{\alpha_2 + \alpha_3}{3}n^3}$. Also, note that the confusable codewords will be either in a row or in a column of $\cM_{og}$. Henceforth, the total number of confusable codewords that can be confused by a confusing codeword from $\mathcal{M}_{og}$ in the whole $\mathcal{M}_{og}$ is no more than $2n^6 q^{\frac{n\epsilon_1}{2}}$ with high probability.

Overall, based on the analysis for type-I error and type-II error, we proved that with high probability, there are at most $n^6 + 2n^6 q^{\frac{n\epsilon_1}{2}}$ confusable codewords in $\mathcal{M}_{og}$. We can conclude that the probability of James' attack succeeding in confusing the actually transmitted codeword is at most $\frac{n^6 + 2n^6 q^{\frac{n\epsilon_1}{2}}}{q^{n\epsilon_1}} \approx \frac{2n^6}{q^{n\epsilon_1/2}}$, which is still exponentially small in $n$.

\subsection{Converse}
\label{sec:converse}
We now argue that the rate achieved by our scheme is optimal by providing jamming strategies for James that ensure that any rate higher than those attained by our schemes must result in a non-vanishing probability of error.

The rate converse of $\cut-z_w$ follows by standard cutset arguments; since there are at most $\cut n$ symbols on the min-cut, if James corrupts $z_w n$ of these symbols with random noise, the residual throughput is at most $(\cut-z_w)n$. 
The tighter rate converse of $(\cut-2z_w)^{+}$ corresponding to the strong adversary regime of $ z_{ro}+2(z_{rw}+z_{wo}) \geq \cut $ follows from the techniques in \cite{converse}. In~\cite{converse}, even for the special case of a ``parallel-edge network" comprising simply of $\cut$ edges linking the source to the sink, the jammer proceeds as follows. Roughly speaking, James first observes the transmissions on the $z_{ro}$ read-only links, picks a codeword $X'$  uniformly at random from all possible codewords compatible with these observations, and then on the $z_{wo}+ z_{rw}$ links he can write on, he replaces the transmissions with the transmissions corresponding to $X'$. It can then be seen that the decoder Bob is unable to determine whether Alice's actual transmission is $X$ or $X'$. In~\cite{converse} this proof is formalized by combining with a Fano's inequality-based argument to show that the probability of error is bounded away from zero for any code (including scenarios wherein stochastic encoding is employed). 
When specialized to the specific subspace codes used in this work, James' attack proceeds as follows.  If the rate exceeds $\cut - 2z_w$, then by the strong-adversary condition he can pick a codeword $V({X}')$ compatible with his observations on his read-only links, and then replace the transmissions on the $z_w$ he can corrupt to be compatible with this $V({X}')$.

\subsection{Secrecy capacity}\label{sec:coset_code}
We explain how to use the coset code of~\cite{el2012secure} to securely achieve the rate $\rate_{sec}^\ast = \cut - z_r-z_w - \epsilon$ for the weak adversary regime, i.e., $\cut>z_{ro}+2z_w$. Recall that the capacity for the strong adversary regime is equal to zero.

The coding strategy consists of using an MDS code $\cC_{\text{MDS}} \in \mathbb{F}_{q^n}^{\rate}$ with length equal to $\rate$ and dimension equal to $z_r$. Alice partitions the space $\mathbb{F}_{q^n}^{\rate}$ into $q^{n(\rate - z_r)}$ cosets of $\cC_{\text{MDS}}$, each of size $q^{nz_r}$. Let $H$ be the $(\rate -z_r)\times \rate$ parity check of $\cC_{\text{MDS}}$. To send a message $m$ that consists of $n(\rate-z_r)$ symbols, Alice chooses a vector $s$ uniformly at random from the coset $\cC_1\subset\mathbb{F}_{q^n}^\rate$ of $\cC_{\text{MDS}}$ that satisfies $m = Hs$ for all $s\in \cC_1$. Due to the properties of MDS codes, observing any $nz_r$ or less symbols of $s$ does not reveal any information about $m=Hs$. To send $s$ over the network, Alice represents $s$ as a vector in $\mathbb{F}_q^{n\rate}$ and encodes it using the random subspace code introduced in this paper. Perfect secrecy is maintained because James observes at most $z_r$ links which reveal at most $nz_r$ symbols of $s$. Reliability against James' jamming attack is ensured with high probability by the subspace codes. After receiving and decoding $s$ with high probability, Bob simply computes $m = Hs$ to recover the message.

Note that the network performs linear operations on the transferred packets which may give more information than intended to the adversary and break the perfect secrecy. The authors of~\cite{el2012secure} show that perfect secrecy is guaranteed under any random linear network code as long as all transfer matrices $T_{z_r}$ of any collection of $z_r$ links do not belong to the space spanned by the rows of $H$. This can be maintained by choosing a large enough field size $q$.

Reliable communication is guaranteed as long as $\rate = \cut - z_w - \epsilon$ for sufficiently small values of $\epsilon$. This ensures a secure and reliable transmission of a message $m$ of length $n(\rate - z_r)$ and therefore achieves $\rate_{sec} = \cut - z_r - z_w - \epsilon$.

The fact that, for all parameter regimes in $\cut, z_r, z_w$, communication is simultaneously reliable and secure (even if one only requires weak secrecy) follows from the corresponding converse argument for parallel-edge networks with overwrite adversaries in~\cite{converse}.




\section{Acknowledgement}
\label{sec:ack}

RB's work was supported by the European Research Council (ERC) under the European Union’s Horizon 2020 research and innovation programme (grant agreement No. 801434) and from the Technical University of Munich - Institute for Advanced Studies, funded by the German Excellence Initiative and European Union Seventh Framework Programme under Grant Agreement No. 291763.
SJ's work was supported by funding from the Hong Kong UGC GRF grants 14304418, 14300617 and 14313116.
YZ has received funding from the European Union’s Horizon 2020 research and innovation programme under grant agreement No 682203-ERC-[Inf-Speed-Tradeoff].

\bibliographystyle{ieeetr}
\bibliography{IEEEabrv,Refs.bib}

\begin{thebibliography}{10}

\bibitem{video}
S.~Li, ``{A short video descriting the paper ``Network Coding with Myopic
  Adversaries" by Sijie Li, Rawad Bitar, Sidharth Jaggi and Yihan Zhang},''
  {\em Zoom}, 2021.
\newblock
  \url{https://cuhk.zoom.us/rec/share/xPwmmO5jVTZxbElDCUoIACY_rcuRqkuq09D5fy4j8xBbtqdtTQaV1o-KohbvJOLS.SjKd6sBzYmKvyiHt?startTime=1612281070000}.

\bibitem{ahlswede2000network}
R.~Ahlswede, N.~Cai, S.-Y. Li, and R.~W. Yeung, ``Network information flow,''
  {\em IEEE Transactions on information theory}, vol.~46, no.~4,
  pp.~1204--1216, 2000.

\bibitem{li2003linear}
S.-Y. Li, R.~W. Yeung, and N.~Cai, ``Linear network coding,'' {\em IEEE
  transactions on information theory}, vol.~49, no.~2, pp.~371--381, 2003.

\bibitem{koetter2003algebraic}
R.~Koetter and M.~M{\'e}dard, ``An algebraic approach to network coding,'' {\em
  IEEE/ACM transactions on networking}, vol.~11, no.~5, pp.~782--795, 2003.

\bibitem{jaggi2005polynomial}
S.~Jaggi, P.~Sanders, P.~A. Chou, M.~Effros, S.~Egner, K.~Jain, and L.~M.
  Tolhuizen, ``Polynomial time algorithms for multicast network code
  construction,'' {\em IEEE Transactions on Information Theory}, vol.~51,
  no.~6, pp.~1973--1982, 2005.

\bibitem{zhang2006hot}
S.~Zhang, S.~C. Liew, and P.~P. Lam, ``Hot topic: Physical-layer network
  coding,'' in {\em Proceedings of the 12th annual international conference on
  Mobile computing and networking}, pp.~358--365, 2006.

\bibitem{dimakis2010network}
A.~G. Dimakis, P.~B. Godfrey, Y.~Wu, M.~J. Wainwright, and K.~Ramchandran,
  ``Network coding for distributed storage systems,'' {\em IEEE transactions on
  information theory}, vol.~56, no.~9, pp.~4539--4551, 2010.

\bibitem{CodedBulk:21}
``Inter-datacenter bulk transfers with codedbulk,'' in {\em 18th {USENIX}
  Symposium on Networked Systems Design and Implementation ({NSDI} 21)},
  {USENIX} Association, Apr. 2021.

\bibitem{yeung2006network}
R.~W. Yeung, N.~Cai, {\em et~al.}, ``Network error correction, i: Basic
  concepts and upper bounds,'' {\em Communications in Information \& Systems},
  vol.~6, no.~1, pp.~19--35, 2006.

\bibitem{cai2006network}
N.~Cai, R.~W. Yeung, {\em et~al.}, ``Network error correction, ii: Lower
  bounds,'' {\em Communications in Information \& Systems}, vol.~6, no.~1,
  pp.~37--54, 2006.

\bibitem{JagLKHME:08}
S.~{Jaggi}, M.~{Langberg}, S.~{Katti}, T.~{Ho}, D.~{Katabi}, M.~{Medard}, and
  M.~{Effros}, ``Resilient network coding in the presence of byzantine
  adversaries,'' {\em IEEE Transactions on Information Theory}, vol.~54, no.~6,
  pp.~2596--2603, 2008.

\bibitem{silva2011universal}
D.~Silva and F.~R. Kschischang, ``Universal secure network coding via
  rank-metric codes,'' {\em IEEE Transactions on Information Theory}, vol.~57,
  no.~2, pp.~1124--1135, 2011.

\bibitem{silva2010universal}
D.~Silva and F.~R. Kschischang, ``Universal secure error-correcting schemes for
  network coding,'' in {\em 2010 IEEE International Symposium on Information
  Theory}, pp.~2428--2432, IEEE, 2010.

\bibitem{metric}
D.~{Silva} and F.~R. {Kschischang}, ``On metrics for error correction in
  network coding,'' {\em IEEE Transactions on Information Theory}, vol.~55,
  no.~12, pp.~5479--5490, 2009.

\bibitem{KK08}
R.~{Koetter} and F.~R. {Kschischang}, ``Coding for errors and erasures in
  random network coding,'' {\em IEEE Transactions on Information Theory},
  vol.~54, no.~8, pp.~3579--3591, 2008.

\bibitem{YaoSJL:14}
H.~{Yao}, D.~{Silva}, S.~{Jaggi}, and M.~{Langberg}, ``Network codes resilient
  to jamming and eavesdropping,'' {\em IEEE/ACM Transactions on Networking},
  vol.~22, no.~6, pp.~1978--1987, 2014.

\bibitem{hayashi2017secrecy}
M.~Hayashi, M.~Owari, G.~Kato, and N.~Cai, ``Secrecy and robustness for active
  attack in secure network coding,'' in {\em 2017 IEEE International Symposium
  on Information Theory (ISIT)}, pp.~1172--1176, IEEE, 2017.

\bibitem{converse}
Q.~{Zhang}, S.~{Kadhe}, M.~{Bakshi}, S.~{Jaggi}, and A.~{Sprintson}, ``Talking
  reliably, secretly, and efficiently: A “complete” characterization,''
  {\em 2015 IEEE Information Theory Workshop}, pp.~1--5, 2015.

\bibitem{dey2019sufficiently}
B.~K. Dey, S.~Jaggi, and M.~Langberg, ``Sufficiently myopic adversaries are
  blind,'' {\em IEEE Transactions on Information Theory}, vol.~65, no.~9,
  pp.~5718--5736, 2019.

\bibitem{el2012secure}
S.~El~Rouayheb, E.~Soljanin, and A.~Sprintson, ``Secure network coding for
  wiretap networks of type {II},'' {\em IEEE Transactions on Information
  Theory}, vol.~58, no.~3, pp.~1361--1371, 2012.

\bibitem{pawar2011securing}
S.~Pawar, S.~El~Rouayheb, and K.~Ramchandran, ``Securing dynamic distributed
  storage systems against eavesdropping and adversarial attacks,'' {\em IEEE
  Transactions on Information Theory}, vol.~57, no.~10, pp.~6734--6753, 2011.

\bibitem{shah2013secure}
N.~B. Shah, K.~Rashmi, and K.~Ramchandran, ``Secure network coding for
  distributed secret sharing with low communication cost,'' in {\em 2013 IEEE
  International Symposium on Information Theory}, pp.~2404--2408, IEEE, 2013.

\bibitem{wang2018epsilon}
Q.~Wang, H.~Sun, and M.~Skoglund, ``The $\epsilon$-error capacity of symmetric
  pir with byzantine adversaries,'' in {\em 2018 IEEE Information Theory
  Workshop (ITW)}, pp.~1--5, IEEE, 2018.

\bibitem{yu2019lagrange}
Q.~Yu, S.~Li, N.~Raviv, S.~M.~M. Kalan, M.~Soltanolkotabi, and S.~A.
  Avestimehr, ``Lagrange coded computing: Optimal design for resiliency,
  security, and privacy,'' in {\em The 22nd International Conference on
  Artificial Intelligence and Statistics}, pp.~1215--1225, PMLR, 2019.

\bibitem{yang2013network}
Y.~Yang, T.~Ho, and W.~Huang, ``Network error correction with limited feedback
  capacity,'' {\em arXiv preprint arXiv:1312.3823}, 2013.

\bibitem{ho2006random}
T.~Ho, M.~M{\'e}dard, R.~Koetter, D.~R. Karger, M.~Effros, J.~Shi, and
  B.~Leong, ``A random linear network coding approach to multicast,'' {\em IEEE
  Transactions on Information Theory}, vol.~52, no.~10, pp.~4413--4430, 2006.

\bibitem{tian2016arbitrarily}
P.~Tian, S.~Jaggi, M.~Bakshi, and O.~Kosut, ``Arbitrarily varying networks:
  Capacity-achieving computationally efficient codes,'' in {\em 2016 IEEE
  International Symposium on Information Theory (ISIT)}, pp.~2139--2143, IEEE,
  2016.

\bibitem{charles2006signatures}
D.~Charles, K.~Jain, and K.~Lauter, ``Signatures for network coding,'' in {\em
  2006 40th Annual Conference on Information Sciences and Systems},
  pp.~857--863, IEEE, 2006.

\bibitem{zhao2007signatures}
F.~Zhao, T.~Kalker, M.~M{\'e}dard, and K.~J. Han, ``Signatures for content
  distribution with network coding,'' in {\em 2007 IEEE International Symposium
  on Information Theory}, pp.~556--560, IEEE, 2007.

\bibitem{cai2002secure}
N.~Cai and R.~W. Yeung, ``Secure network coding,'' in {\em Proceedings IEEE
  International Symposium on Information Theory,}, p.~323, IEEE, 2002.

\bibitem{yao2012passive}
H.~Yao, S.~Jaggi, and M.~Chen, ``Passive network tomography for erroneous
  networks: A network coding approach,'' {\em IEEE Transactions on Information
  Theory}, vol.~58, no.~9, pp.~5922--5940, 2012.

\bibitem{noga2004probabilistic}
N.~Alon and J.~H. Spencer, {\em The probabilistic method}.
\newblock John Wiley \& Sons, 2004.

\bibitem{ben2018prob}
B.~Doerr, ``Probabilistic tools for the analysis of randomized optimization
  heuristics,'' {\em CoRR}, vol.~abs/1801.06733, 2018.

\end{thebibliography}




\end{document}